\documentclass[a4paper,UKenglish,cleveref, autoref, thm-restate]{lipics-v2021}
\pdfoutput=1

\usepackage{amsfonts}
\usepackage{amsthm}
\usepackage{mathtools}
\usepackage{appendix}
\usepackage[disable]{todonotes}
\usepackage{lipsum}
\usepackage{wrapfig}

\usepackage{thm-restate}

\usepackage{pgf,tikz,pgflibraryarrows,pgffor,pgflibrarysnakes}
\tikzstyle{background}=[rectangle,fill=gray!10, inner sep=0.1cm, rounded corners=0mm]
\usepgflibrary{shapes}
\usetikzlibrary{snakes,automata}

\usepackage{float}
\usepackage{algorithm}
\usepackage{algpseudocode}

\newcommand{\calR}{\mathcal{R}}
\newcommand{\proj}[1]{\pi_{#1}}
\newcommand{\varSet}{\mathcal{X}}
\usepackage{stmaryrd}
\newcommand{\intInterv}[2]{\llbracket #1 , #2 \rrbracket}
\newcommand{\expr}[1]{\textup{\textsf{Exp}}(#1)}
\newcommand{\linExpr}[1]{\textup{\textsf{Exp}}_\ell(#1)}
\newcommand{\affExpr}[1]{\textup{\textsf{Exp}}_a(#1)}
\newcommand{\subs}[1]{\textup{\textsf{Sub}}(#1)}

\newcommand{\valuations}[1]{\textup{\textsf{Val}}(#1)}
\newcommand{\outFctCRA}{o}
\newcommand{\linPart}[1]{#1^{(l)}}
\newcommand{\affPart}[1]{#1^{(a)}}
\newcommand{\expToMap}[1]{\underline{#1}}
\newcommand{\noUnderline}[1]{#1}
\newcommand{\vectSet}[2]{#1^{1 \times #2}}
\newcommand{\matrSet}[3]{#1^{#2 \times #3}}
\newcommand{\sqmatrSet}[2]{#1^{#2 \times #2}}
\newcommand{\genMono}[1]{\left< #1 \right>}
\newcommand{\linSpan}[1]{\textup{\textsf{span}}\left(#1\right)}
\newcommand{\affSpan}[1]{\textup{\textsf{aff}}\left(#1\right)}
\newcommand{\lReachSet}[1]{\textup{\textsf{LR}}\left(#1\right)}
\newcommand{\rReachSet}[1]{\textup{\textsf{RR}}\left(#1\right)}
\newcommand{\linClosure}[1]{\overline{#1}^\ell}
\newcommand{\affClosure}[1]{\overline{#1}^a}

\newcommand{\linHull}[1]{\linClosure{\lReachSet{#1}}}
\newcommand{\affHull}[1]{\affClosure{\lReachSet{#1}}}

\newcommand{\chgBaseMatr}[2]{P_{#1 \to #2}}
\newcommand{\resizIdMatr}[2]{I_{#1 \times #2}}
\newcommand{\canonBasis}[1]{E_{#1}}
\newcommand{\im}[1]{\textup{\textsf{im}}\left(#1\right)}
\renewcommand{\ker}[1]{\textup{\textsf{ker}}\left(#1\right)}
\usepackage{xspace}
\newcommand{\LH}{strongest Z-linear invariant\xspace}
\newcommand{\AH}{strongest Z-affine invariant\xspace}
\newcommand{\LAH}{strongest Z-linear/affine invariant\xspace}
\newcommand{\LrespAH}{strongest Z-linear (\resp Z-affine) invariant\xspace}

\newcommand{\CRA}{\ensuremath{\mathrm{CRA}}\xspace}
\newcommand{\WA}{\ensuremath{\mathrm{WA}}\xspace}
\newcommand{\eg}{\textit{e.g.~}}
\newcommand{\ie}{\textit{i.e.~}}
\newcommand{\ea}{\textit{et al. }}
\newcommand{\resp}{resp. }

\newcommand{\permutMatrSet}[1]{\textup{\textsf{Perm}}_{#1}}

\newcommand{\set}[1]{\left\{#1\right\}}

\newcommand{\nptime}{\textsc{NPTime}\xspace}
\newcommand{\pspace}{\textsc{PSpace}\xspace}
\newcommand{\exptime}{\textsc{ExpTime}\xspace}
\newcommand{\nexptime}{\textsc{NExpTime}\xspace}

\newcommand{\texptime}{\textsc{2-ExpTime}\xspace}

\newcommand{\dWA}{d} 
\newcommand{\dInv}{k} 
\newcommand{\lInv}{n} 
\newcommand{\rCRA}{\dInv} 
\newcommand{\sCRA}{\lInv} 

\title{\hbox{Minimizing Cost Register Automata over a Field}}

\author{Yahia Idriss {Benalioua}}{Aix Marseille Univ, CNRS, LIS, Marseille, France}{yahia-idriss.benalioua@lis-lab.fr}{}{}

\author{Nathan Lhote}{Aix Marseille Univ, CNRS, LIS, Marseille, France}{nathan.lhote@lis-lab.fr}{}{}

\author{Pierre-Alain Reynier}{Aix Marseille Univ, CNRS, LIS, Marseille, France}{pierre-alain.reynier@lis-lab.fr}{}{}

\authorrunning{Y.I. Benalioua, N. Lhote and P.-A. Reynier} 

\Copyright{Yahia Idriss {Benalioua}, Nathan Lhote and Pierre-Alain Reynier} 



\category{} 

\relatedversion{This is the full version of a paper accepted at MFCS 2024} 

\funding{This work has been partly funded by the QuaSy project (ANR-23-CE48-0008).}

\EventEditors{Rastislav Kr\'{a}lovi\v{c} and Anton\'{i}n Ku\v{c}era}
\EventNoEds{2}
\EventLongTitle{49th International Symposium on Mathematical Foundations of Computer Science (MFCS 2024)}
\EventShortTitle{MFCS 2024}
\EventAcronym{MFCS}
\EventYear{2024}
\EventDate{August 26--30, 2024}
\EventLocation{Bratislava, Slovakia}
\EventLogo{}
\SeriesVolume{306}
\ArticleNo{16}

\nolinenumbers
\hideLIPIcs

\begin{document}

\maketitle

\begin{abstract}
Weighted automata (\WA) are an extension of finite automata that define functions
 from words to values in a given semiring.
 An alternative deterministic model, called Cost Register Automata ($\CRA$),
 was introduced by Alur \ea It enriches deterministic finite automata with
 a finite number of registers, which store values, updated at each transition
 using the operations of the semiring.
It is known that \CRA with register updates defined by linear maps
 have the same expressiveness as \WA.
 Previous works have studied the register minimization problem:
 given a function computable by a \WA and an integer $k$,
 is it possible to realize it using a \CRA with at most $k$ registers?

 In this paper, we solve this problem for $\CRA$ over a field
 with linear register updates, using the notion of linear hull,
 an algebraic invariant of $\WA$
 introduced recently by Bell and Smertnig.
We then generalise the approach to solve a more
challenging problem, that consists in minimizing simultaneously the number
of states and that of registers.
 In addition, we also lift our results to the setting of \CRA with affine updates.
Last, while the linear hull was recently shown to be computable by Bell and Smertnig, no complexity bounds were given.
 To fill this gap, we provide two new algorithms to compute invariants of \WA. This allows us to
 show that the register (resp. state-register) minimization problem can be solved in 2-\exptime
 (resp. in \nexptime).
\end{abstract}

\section{Introduction}

\textbf{Weighted automata (\WA)}
 are a quantitative extension of finite state automata and have been studied since the sixties~\cite{Schutzenberger61b}. These automata define functions from words to a given semiring: each transition
has a weight in the semiring and the weight of an execution is the product of the weights of the
transitions therein; the non-determinism of the model is handled using
the sum of the semiring: the weight associated with a word is the sum of
the weights of the different executions over this word. Functions realized by weighted
automata are called rational series.
This fundamental model has been widely studied
during the last decades~\cite{HBWA}.
While some expressiveness results can be
obtained in a general framework (such as the equivalence with
rational expressions), the decidability status of important problems
heavily depends on the considered semiring. Amongst the classical problems of interest,
one can mention \emph{equivalence}, \emph{sequentiality} (resp. \emph{unambiguity}), which aims at determining whether there exists
an equivalent deterministic (resp. unambiguous) \WA, and \emph{minimization}, which aims at minimizing the number of states.

Weighted automata over a field (\eg the field of rationals $\mathbb{Q}$) enjoy many nice properties: the equivalence of weighted
automata is decidable and they can be minimized, and both can be done efficiently
(see \eg \cite[Theorem 4.10 and Corollary 4.17 (Chapter III)]{Sakarovitch09}). The sequentiality and
unambiguity
are also
decidable, as shown  recently in~\cite{BellS21,BS23}, with no complexity bounds however.
The most studied semirings which are not fields are the tropical semirings
and the semiring of languages, and in both cases equivalence is undecidable
(see~\cite[Section 3]{Daviaud20} and~\cite[Theorem 8.4]{Berstel79})
and no minimization algorithm is known. Regarding sequentiality, partial
decidability results have been obtained for these semirings
using the notion of twinning property~\cite{Choffrut77,Mohri97}.

\textbf{Cost register automata (\CRA)} have been introduced more recently by Alur \ea \cite{AlurDDRY13}.
A cost register automaton is a deterministic finite state automaton endowed with a finite number of
registers storing values from the semiring.
The registers are initialized by some values, then at each
transition the values are updated using the operations and constants of the semiring.
Several fragments of \CRA can be considered by restricting the operations allowed.
For instance, an easy observation
is that \WA are exactly \CRA with one state (however, one can observe that adding states does not extend expressiveness) and linear updates, \emph{i.e.}
updates of the form $X \coloneqq \sum_{i=1}^k X_i * c_i$
(intuitively, the new values of the registers only depend
linearly on the previous ones).
Thus, the model of linear \CRA is an alternative to \WA which
allows to trade non-determinism for registers.

\textbf{The register minimization problem.} As \CRA are finite state automata
extended with registers storing elements from the semiring, it is natural to aim
at minimizing the number of registers used.
For a given class $\mathcal{C}$ of \CRA, this problem asks,
given a \WA and a number $\rCRA$, whether there exists an
equivalent \CRA in $\mathcal{C}$ with at most $\rCRA$ registers.
From a practical point of view,
reducing the number of registers allows to reduce the memory usage, since a register can require unbounded memory.
From a theoretical point of view, this problem can be understood as a refinement of the classical
problem of minimization of WA.
Indeed, a \WA can be translated into a linear CRA with a single state, and as many registers as the number
of states of the \WA.
This problem has been studied in~\cite{DBLP:conf/icalp/AlurR13,DRT16,DaviaudJRV17} for
three different models of \CRA but in all these works,
the additive law of the semiring is not allowed (\emph{i.e.} updates of the form
$X \coloneqq Y+Z$ are forbidden).
It is worth noticing that~\cite{DRT16} encompasses the
case of \CRA over a field, with only updates of the form
$X \coloneqq Y*c$, with $c$ an element of the field.

While the minimal number of registers needed to realise a \WA (also known as the \emph{register complexity}) is upper bounded by the number of states of a minimal \WA,
it may be possible to build an equivalent \CRA with fewer registers, but more states.
Hence there is a \emph{tradeoff} between the number of states and the number of registers.
This leads to the following \emph{state-register minimization problem for \CRA} which asks,
for a class $\mathcal{C}$ of \CRA,
given a \WA and integers $\sCRA,\rCRA$
whether an equivalent \CRA in $\mathcal{C}$ with $\sCRA$ states and $\rCRA$ registers can be constructed.
In this framework, the classical minimization of \WA corresponds to minimizing the number of registers
while using only one state, for the class of linear \CRA.

\textbf{The linear hull.} As mentioned before, the case of fields
is well-behaved to obtain decidability results.
In their recent work~\cite{BellS21},
Bell and Smertnig introduced the notion of \emph{linear hull} of a \WA.
This notion is inspired by the algebraic theory needed to study polynomial automata but cast
into a linear setting.
A linear algebraic set (aka linear Zariski closed set) is a finite union of vector
subspaces: we later call them \emph{Z-linear sets}.
Given a Z-linear set $S = \bigcup_{i=1}^p V_i$, the dimension of $S$ is the maximum of the dimensions of the $V_i$'s. In this work, the size of the union, $p$, is called the
length of $S$.
Observe that such Z-linear sets were also used in~\cite{ColcombetP17}
for a category-theoretic approach to minimization of weighted automata over a field.
We say such a set is an invariant if it contains the initial vector and is
stable under the updates of the automaton.
Then the linear hull of a weighted automaton is the strongest Z-linear invariant.
In~\cite{BellS21}, Bell \& Smertnig show that computing the
linear hull of a minimal automaton allows to decide sequentiality and unambiguity.
In addition, in~\cite{BS23}, they show that
the linear hull can effectively be computed, without providing complexity bounds however.

\textbf{Contributions.} In this work, we deepen the analysis of the linear hull of a \WA
in order to solve the register and state-register minimization problems
for linear \CRA.
In addition, we also provide new algorithms
to compute the linear hull which come with complexity upper bounds, which can
be used to derive complexity results for minimization problems as well as for sequentiality
and unambiguity of \WA.
More precisely, our contributions are as follows:
\begin{itemize}
\item Firstly, we show that \emph{the register minimization problem} for the class
of linear \CRA over a field can be solved in \texptime.
To this end, given a rational series $f$, we show that the minimal number of registers
needed to realize $f$ using a linear \CRA is exactly the dimension of the linear hull
of a minimal \WA of $f$.
We then show that the linear hull of a \WA can be computed in
\texptime.
We show that this complexity drops down to \exptime for the particular case of commuting transition matrices
(which includes the case of a single letter alphabet), with a matching lower bound.

\item As a consequence of the computation of the linear hull of a \WA
and of results proved in~\cite{BellS21}, we obtain a \texptime upper bound for the problems of
\emph{sequentiality and unambiguity of weighted automata} over a field, closing a question
raised in~\cite{BS23}.

\item Secondly, we prove that the \emph{state-register minimization problem}
for linear \CRA can be solved in \nexptime.
More precisely, given a minimal \WA $A$,
we show a correspondence between Z-linear invariants of $A$
and linear \CRA equivalent to $A$.
This correspondence maps the length (resp. dimension) of the
invariant to the number of states (resp. registers) of the equivalent linear \CRA.
We then provide a (constructive) \nexptime algorithm that,
given a minimal \WA and two integers $\sCRA,\rCRA$,
guesses a well-behaved invariant allowing to exhibit a satisfying equivalent \CRA.

\item Last, we actually present these results in a more general setting, by considering \emph{affine}
\CRA, which are a
slight extension of linear \CRA allowing to use affine maps in the updates of the registers.
\end{itemize}

\textbf{Outline of the paper.}
\todo{PA: new}
We present the models of weighted automata and cost register automata in Section~\ref{sec:prelim}.
We then formally define the two problems we consider, \emph{i.e.} register and state-register minimization
problems, and state our main results in
Section~\ref{sec:pbres}. In Section~\ref{sec:characterization}, we introduce the
necessary topological notions to define
Z-linear/Z-affine set and invariants of weighted automata,
and detail our characterizations of the register and state-register
complexities of a rational series. Finally in Section~\ref{sec:algos},
we present our algorithms, as well as their consequences
in terms of decidability and complexity for the two problems we
consider.
Omitted proofs and more details for Sections~\ref{sec:characterization} and~\ref{sec:algos}
can be found in the appendix.

\section{Weighted Automata and Cost Register Automata}
\label{sec:prelim}
\textbf{Basic concepts and notations.}
An alphabet $\Sigma$ is a finite set of letters.
The set of finite words over $\Sigma$ will be denoted by $\Sigma^*$,
the empty word by $\epsilon$ and, for two words $u$ and $v$, $uv$ will denote their concatenation.
For two sets $X$ and $Y$, we denote by $X \times Y$ their cartesian product and by
$\proj{X}\colon X \times Y \to X$ and $\proj{Y}\colon X \times Y \to Y$
we denote the canonical projection on $X$ and $Y$ respectively.
The set nonnegative integers will be denoted by $\mathbb{N}$.
For two integers $i,j$, we will denote by $\intInterv{i}{j}$
the interval of integers between $i$ and $j$ (both included).

A \emph{semigroup} $(S,*)$ is a set $S$ together with an associative binary operation $*$.
If $(S,*)$ has an identity element $e$, $(S,*,e)$ is called a \emph{monoid}
and if, moreover, every element has an inverse, $(S,*,e)$ is called a \emph{group}.
If there is no ambiguity, we will identify algebraic structures with the set that they are defined on.
A semigroup (or a monoid/group) is said to be \emph{commutative} if its law is.
A sub-semigroup (or submonoid/subgroup) of $S$ is a subset of $S$ that is a semigroup (or a monoid/group).
Given $E \subseteq S$, the monoid \emph{generated} by $E$, denoted
$\genMono{E}$, is the smallest sub-monoid of $S$ containing $E$.

A \emph{field} $(\mathbb{K},+,\cdot)$ is a structure where $(\mathbb{K},+,0)$ and
$(\mathbb{K}\setminus \left\{ 0 \right\}, \cdot, 1)$ are commutative groups and multiplication distributes
over addition.
In this work, we will consider $\mathbb{K}$ as the field of rational numbers $\mathbb{Q}$,
or any finite field extension of $\mathbb{Q}$, to perform basic operations in polynomial time.
For all $n \in \mathbb{N}$, $\mathbb{K}^n$ is an $n$-dimensional \emph{vector space}
over the field $\mathbb{K}$.
We will work with row vectors and apply matrices on the right, and we will identify
linear maps (\resp linear forms) with their corresponding matrices (\resp column vectors).
The set of $n$ by $m$ matrices over $\mathbb{K}$ will be denoted by $\matrSet{\mathbb{K}}{n}{m}$,
and $\vectSet{\mathbb{K}}{n}$ (or simply $\mathbb{K}^n$ when there is no ambiguity)
will denote the set of $n$-dimensional vectors.
For any matrix $M$ (\resp vector $v$), and indices $i$ and $j$,
$M_{i,j}$ (\resp $v_i$) will denote the value of the
entry in the $i$-th row and the $j$-th column of $M$ (\resp the $i$-th entry of $v$).
Matrix transposition will be denoted by $M^t$.
A \emph{vector subspace} of $\mathbb{K}^n$ is a subset of $\mathbb{K}^n$
stable by linear combinations and for all subsets $E$ of $\mathbb{K}^n$,
$\linSpan{E}$ will denote the smallest vector subspace of $\mathbb{K}^n$
containing $E$ (if $E$ contains a single vector $(x_1, \dots, x_n)$,
$\linSpan{E}$ will be denoted by $\linSpan{x_1, \dots, x_n}$).

$\mathbb{K}^n$ can also be seen as an $n$-dimensional \emph{affine space}.
Affine maps $f \colon \mathbb{K}^n \to \mathbb{K}^m$ are maps of the form
$f(u) = u \linPart{f} + \affPart{f}$ where $\linPart{f} \in \matrSet{\mathbb{K}}{n}{m}$
and $\affPart{f} \in \vectSet{\mathbb{K}}{m}$.
An \emph{affine subspace} $A$ of $\mathbb{K}^n$ is a subset
of $\mathbb{K}^n$ of the form $A = p + V$ with $p \in A$ and $V$ a vector subspace of $\mathbb{K}^n$.
They are stable by affine combinations (linear combinations with coefficients adding up to 1).
For all $E \subseteq \mathbb{K}^n$, $\affSpan{E}$ will denote the smallest
affine subspace of $\mathbb{K}^n$ containing $E$.

\textbf{Weighted Automata.}
Let $\Sigma$ be a finite alphabet and $(\mathbb{K},+,\cdot)$ be a field.

\begin{definition}[Weighted Automaton]
  A \emph{Weighted Automaton} (\WA for short) of dimension $\dWA$,
  on $\Sigma$ over $\mathbb{K}$, is a triple $\mathcal{R} = (u,\mu,v)$,
  where $u \in \vectSet{\mathbb{K}}{\dWA}$, $v \in \matrSet{\mathbb{K}}{\dWA}{1}$
  and $\mu \colon \Sigma^* \to \sqmatrSet{\mathbb{K}}{\dWA}$ is a monoid morphism.
  We will call $u$ and $v$ the \emph{initial} and \emph{terminal} vectors respectively and
   $\mu(a)$, for $a \in \Sigma$, will be called a \emph{transition matrix}.
  A \WA realizes a \emph{formal power series} over $\Sigma^*$ with coefficients in $\mathbb{K}$
  (a function from $\Sigma^*$ to $\mathbb{K}$) defined, for all $w \in \Sigma^*$,
  by $ \left\llbracket \mathcal{R} \right\rrbracket (w) = u \mu(w) v$.
  Any series that can be realized by a \WA will be called \emph{rational}.
\end{definition}


\begin{wrapfigure}{r}{.35\textwidth}
\vspace{-.5cm}
  \centering
  \scalebox{.9}{
    \begin{tikzpicture}[->,>=stealth',shorten >=1pt,auto,node distance=1.8cm]
      \tikzstyle{every state}= [minimum size=5mm]

      \node[state] (p) at (0,0) {$q_1$};

      \node[state,draw=none] (ps) at (-1,.3) {};
      \node[state,draw=none] (pp) at (-1,-.3) {};

      \node[state] (q) [right of=p] {$q_2$};

      \draw
      (p) edge [bend left] node{$a : 2$}(q)
      (q) edge [bend left] node{$a : 2$}(p)
      (ps) edge [pos=.1] node {$1$}(p)
      (p) edge [pos=.8] node{$1$}(pp);
    \end{tikzpicture}
  }
  \caption{The WA of Example~\ref{ex:WA}.}
  \label{fig:WA}
\end{wrapfigure}
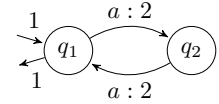

\WA also have a representation in terms of finite-state automata,
in which transitions are equipped with weights. We then say that a \WA
is sequential (resp. unambiguous) when its underlying automaton is.
Formally, we say that a \WA $\mathcal{R} = (u,\mu,v)$ is \emph{sequential} when
$u$ has a single non-zero entry and, for each letter $a$, and each index $i$,
there is at most one index $j$ such that  $\mu(a)_{i,j} \neq 0$.

\begin{example}
  \label{ex:WA}
  We consider the \WA, on the alphabet $\{a\}$ and over the field of real
  numbers, $\mathcal{R} =(u,\mu,v)$ with $u=(1, 0)$, $v = (1, 0)^{t}$, and
  $\mu(a) = \begin{pmatrix}
              0 & 2 \\
              2 & 0
  \end{pmatrix}$.
  One can verify that the function realized by this \WA maps the
  word $a^n$ to $2^n$ if $n$ is even, and to $0$ otherwise.
  It can be represented graphically by the automaton depicted on Figure~\ref{fig:WA}.
\end{example}

A $\WA$ realizing a rational series $f$ is said to be \emph{minimal}
if its dimension is minimal among all the $\WA$ realizing $f$.
We also have the following characterization of minimal \WA
(see~\cite[Proposition 4.8 (Chapter III)]{Sakarovitch09}):
\begin{proposition}
  \label{prop:caracMinRep}
  Let $\mathcal{R} = (u,\mu,v)$ be a $\dWA$-dimensional \WA
  and let $\lReachSet{\mathcal{R}} = u \mu(\Sigma^*) = \left\{ u \mu(w) \,\middle|\, w \in \Sigma^* \right\}$
  be its (left) \emph{reachability set} and $\rReachSet{\mathcal{R}} = \mu(\Sigma^*) v$
  be its right reachability set.

  $\mathcal{R}$ is a minimal \WA if and only if
  $\linSpan{\lReachSet{\mathcal{R}}} = \vectSet{\mathbb{K}}{\dWA}$
  and $\linSpan{\rReachSet{\mathcal{R}}} = \matrSet{\mathbb{K}}{\dWA}{1}$.
\end{proposition}

\textbf{Expressions, substitutions and Cost Register Automata.}
For a field $(\mathbb{K},+,\cdot)$ and a finite set of variables $\varSet$
disjoint from $\mathbb{K}$,
let $\expr{\varSet}$ denote the set of expressions generated by the grammar
$e \Coloneqq k \,|\, X \,|\, e + e \,|\, e \cdot e$, where $k \in \mathbb{K}$
and $X \in \varSet$.
A \emph{substitution} over $\varSet$ is a map $s \colon \varSet \to \expr{\varSet}$.
It can be extended to a map $\expr{\varSet} \to \expr{\varSet}$ by substituting each variable $X$ in the
expression given as an input by $s(X)$.
By identifying $s$ with its extension, we can compose substitutions.
We call \emph{valuations} the substitutions of the form $v \colon \varSet \to \mathbb{K}$.
The set of substitutions over $\varSet$ will be denoted by $\subs{\varSet}$
and the set of valuations $\valuations{\varSet}$.

\begin{definition}[Cost Register Automaton]
  A \emph{cost register automaton} ($\CRA$ for short), on the alphabet $\Sigma$ over the field $\mathbb{K}$,
  is a tuple $\mathcal{A} = (Q, q_0, \varSet, v_0, \outFctCRA, \delta)$
  where $Q$ is a finite set of \emph{states}, $q_0 \in Q$ is the initial state,
  $\varSet$ is a finite set of registers (variables),
  $v_0 \in \valuations{\varSet}$ is the registers' initial valuation,
  $\outFctCRA \colon Q \to \expr{\varSet}$ is the output function,
  and $\delta \colon Q \times \Sigma \to Q \times \subs{\varSet}$ is the transition function.
  We will denote by $\delta_Q \coloneqq \proj{Q} \circ \delta$ the transition function of the underlying
  automaton of the $\CRA$ and $\delta_{\varSet} \coloneqq \proj{\subs{\varSet}} \circ \delta$
  its register update function.

  $\mathcal{A}$ computes a function $\llbracket \mathcal{A} \rrbracket \colon \Sigma^* \to \mathbb{K}$
  defined as follows: the configurations of $\mathcal{A}$ are pairs $(q,v) \in Q \times \valuations{
    \varSet}$.
  The run of $\mathcal{A}$ on a word $w = a_1 \dots a_n \in \Sigma^*$ is the sequence of
  configurations $(q_i,v_i)_{i \in \intInterv{0}{n}}$
  where, $q_0$ is the initial state, $v_0$ is the initial valuation and,
  for all $i \in \intInterv{1}{n}$, $q_i = \delta_Q (q_{i-1},a_i)$
  and $v_i = v_{i-1} \circ \delta_{\varSet}(q_{i-1},a_i)$.
  We then define $\llbracket \mathcal{A} \rrbracket (w) = v_n (\outFctCRA(q_n))$.

  $\delta$ can be extended to words by setting, for all $q \in Q$,
  $\delta (q,\epsilon) = (q , id_{\varSet})$, where $id_{\varSet}$
  is the substitution such that $id_{\varSet}(X) = X$ for all $X \in \varSet$,
  and, for all $a \in \Sigma$ and $w \in \Sigma^*$,
  $\delta_Q (q,aw) = \delta_Q(\delta_Q(q,a),w)$
  and $\delta_{\varSet} (q,aw) = \delta_{\varSet}(q,a)
  \circ \delta_{\varSet}(\delta_Q(q,a),w)$.
  We then have
  \[\llbracket \mathcal{A} \rrbracket (w)
  = v_0 \circ \delta_{\varSet} (q_0, w) (\outFctCRA(\delta_Q(q_0,w)))\]

\end{definition}

\begin{example}[Example~\ref{ex:WA} continued]
  \label{ex:CRA}
  Two \CRA are depicted on Figure~\ref{fig:CRA}.
  They are both on the alphabet
  $\{a\}$ and over the field of real numbers, and both realize the same function as the $\WA$ considered in
  Example~\ref{ex:WA}.
\end{example}


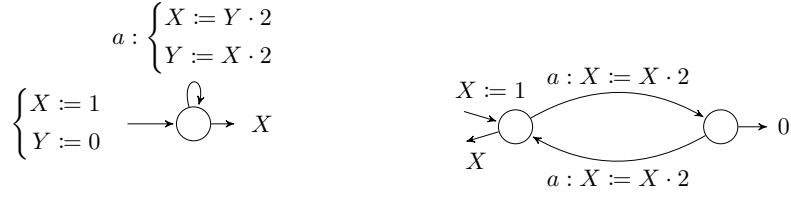
\begin{figure}
  \centering
  \scalebox{.9}{
    \begin{tikzpicture}[->,>=stealth',shorten >=1pt,auto,node distance=1.8cm]
      \tikzstyle{every state}= [minimum size=5mm]

      \node[state] (p) at (0,0) {};

      \node[state,draw=none] (ps) at (-2,0) {$\left \{\begin{aligned}
                                                        X \coloneqq 1\\Y \coloneqq 0
      \end{aligned}\right. $};

      \node[state,draw=none] (pp) at (1,0) {$X$};

      \draw
      (p) edge [loop above] node {$a :\left \{\begin{aligned}
                                                X \coloneqq Y \cdot 2\\Y \coloneqq X \cdot 2
      \end{aligned}\right.$} (p)
      (ps) edge (p)
      (p) edge (pp);
    \end{tikzpicture}
  }
  \hfil
  \scalebox{.9}{
    \begin{tikzpicture}[->,>=stealth',shorten >=1pt,auto,node distance=3cm]
      \tikzstyle{every state}= [minimum size=5mm]

      \node[state] (p) at (0,1) {};

      \node[state,draw=none] (ps) at (-1,1.3) {};
      \node[state,draw=none] (pp) at (-1,.7) {};

      \node[state, accepting right, accepting text = $0$] (q) [right of=p] {};

      \draw
      (p) edge [bend left] node{$a : X \coloneqq X \cdot 2$}(q)
      (q) edge [bend left] node{$a : X \coloneqq X \cdot 2$}(p)
      (ps) edge [pos=-.5] node {$X\coloneqq 1$}(p)
      (p) edge [pos=1.2] node{$X$}(pp);
    \end{tikzpicture}
  }

  \caption{Two CRA detailed in Example~\ref{ex:CRA}. Registers are denoted
  by letters $X,Y$.}
  \label{fig:CRA}
\end{figure}

An expression is called \emph{linear} if it has the form
$\sum_{i=1}^{\rCRA} \alpha_i X_i$, for some family of $\alpha_i \in \mathbb{K}$ and $X_i \in \varSet$,
and if it has the form $\sum_{i=1}^{\rCRA} \alpha_i X_i + \beta$, for some $\beta \in \mathbb{K}$,
it is called \emph{affine}.
We will denote by $\linExpr{\varSet}$ (\resp $\affExpr{\varSet}$) the set of linear
(\resp affine) expressions.

\begin{definition}[Linear/Affine \CRA]
  A $\CRA$ $\mathcal{A} = (Q, q_0, \varSet, \allowbreak v_0, \outFctCRA, \delta)$ is called \emph{linear}
  if, $\delta_{\varSet}(q,a)(X) \in \linExpr{\varSet}$ and $\outFctCRA(q)\in \linExpr{\varSet}$,
  for all $q \in Q, a \in \Sigma$ and $X \in \varSet$,
  and if $\delta_{\varSet}(q,a)(X) \in \affExpr{\varSet}$ and $\outFctCRA(q)\in \affExpr{\varSet}$,
  the \CRA is called \emph{affine}.
\end{definition}

Linear \CRA are a particular case of affine \CRA and, given an affine \CRA
it is always possible to define an equivalent linear \CRA using one more register
with a constant value of $1$ to realize affine register updates in a linear way, thus :
\begin{remark}
  Linear and affine \CRA have the same expressiveness.
\end{remark}
The added cost of a register will however become relevant when we will consider
minimization problems in the next sections.

Observe that we can assume that $\varSet=\left\{X_1, \dots, X_\rCRA \right\}$ is ordered,
and identify any linear expression $e =  \sum_{i=1}^{\rCRA} \alpha_i X_i$
(with the $\alpha_i$ not present in the expression assumed to be 0)
with the linear form $\expToMap{e} \colon \mathbb{K}^\rCRA \to \mathbb{K}$
defined by the column vector $(\alpha_1, \dots, \alpha_\rCRA)^t$.
We can then identify any linear substitution $s \colon \varSet \to \linExpr{\varSet}$
with the linear map $\expToMap{s} \colon \mathbb{K}^\rCRA \to \mathbb{K}^\rCRA$
defined by the block matrix $(\expToMap{s(X_1)} | \cdots | \expToMap{s(X_\rCRA)})$,
and we can identify any valuation $v \colon \varSet \to \mathbb{K}$ with the vector
$\expToMap{v}=(v(X_1), \cdots, v(X_\rCRA))$ of the vector space $\mathbb{K}^\rCRA$.

In the following, we will drop the underline notation and make the identifications implicitly.

Thanks to these observations, the registers of a linear
\CRA and their updates can be characterized by the values of the vector
associated with $\noUnderline{v_0}$, and the linear maps associated with the
$\noUnderline{\delta_{\mathcal{X}}(q,a)}$ and $\noUnderline{\outFctCRA(q)}$,
for all $q \in Q$ and $a \in \Sigma$, and we can check that
\[
  \llbracket \mathcal{A} \rrbracket (w) = \noUnderline{v_0}\ \noUnderline{\delta_{\varSet}(q_0,w)}\
  \noUnderline{\outFctCRA(\delta_Q(q_0,w))}
\]

We can also identify affine expressions with affine forms and affine substitutions with affine maps
to simplify dealing with affine \CRA.
We will define and use these identifications in Appendix~\ref{apx:proofChar}.

\begin{proposition}[\cite{AlurDDRY13}]
  \label{prop:linRepEquivCRA1Stt}
  There is a bijection between \WA and linear \CRA with a single state.
\end{proposition}

Given a \WA, one can build an equivalent \CRA with as many registers as states of the \WA:
for each letter $a$, the transition matrix $\mu(a)$ can be interpreted as a (linear) substitution,
associated with the self-loop of label $a$.
The converse easily follows from the previous observations when the \CRA has a single state.

\begin{example}[Example~\ref{ex:WA} continued]
  The \CRA depicted on the left of Figure~\ref{fig:CRA} is obtained by the translation of
  the \WA of Figure~\ref{fig:WA} into \CRA with a single state.
\end{example}

\begin{remark}\label{rk:seq}Sequential \WA are exactly linear \CRA with a single register.
\end{remark}

Indeed, both sequential \WA and linear \CRA with only one register are
deterministic finite automata that can also store a single value updated at each
transition using only products.
They can then be identified.

\section{Problems and Main Results}
\label{sec:pbres}

\begin{definition}[Register minimization problem]
  \label{def:min-prob}
  Given a class $\mathcal{C}$ of \CRA, we ask:
  \begin{itemize}
    \item \textbf{Input:} a rational series $f$ realized by a given \WA,
    and an integer $\rCRA \in \mathbb{N}$
    \item \textbf{Question:} Does there exist a \CRA realizing $f$
    in the class $\mathcal{C}$ with at most $\rCRA$ registers?
  \end{itemize}
\end{definition}

We will show this problem is decidable for the classes of linear and affine \CRA:
\begin{restatable}{theorem}{thmRegMin}
\label{thm:reg-min}
The register minimization problem is decidable for the classes of linear
and affine \CRA
in \texptime.
Furthermore, the algorithm exhibits a solution when it exists.
\end{restatable}

For a rational series $f$, the minimal number of registers needed
to realize $f$ using \CRA in some class $\mathcal{C}$ 
is called its \emph{register complexity with respect to class $\mathcal{C}$}.
\todo{PA: new}
Dually, if one wants to minimize the number of states, then we know we can always build
a linear (hence affine) \CRA with a single state (Proposition~\ref{prop:linRepEquivCRA1Stt}).
A more ambitious goal is to try to reduce simultaneously the number of states 
and of registers, in some given class $\mathcal{C}$ of \CRA.
Observe that, in general, there is no \CRA with minimal numbers of both states 
and registers (see Example~\ref{ex:CRA}).
Given a rational series $f$, we say that a pair $(\sCRA,\rCRA)$ is \emph{optimal} if
$f$ can be realized by a \CRA
in class $\mathcal{C}$ with $\sCRA$ states and $\rCRA$ registers
and no \CRA of $\mathcal{C}$ realizing $f$ with at most $\sCRA$ states can have strictly
less than $\rCRA$ registers and vice-versa.

Formally, we call the \emph{state-register complexity with respect to class $\mathcal{C}$}
of a rational series $f$,
the set of optimal pairs of integers $(\sCRA,\rCRA)$.

This leads to the definition of a second minimization problem:
\begin{definition}[State-Register minimization problem]
  Given a class $\mathcal{C}$ of \CRA, we ask:
  \begin{itemize}
    \item \textbf{Input:} a rational series $f$ realized by a given \WA,
    and two integers $\sCRA, \rCRA \in \mathbb{N}$
    \item \textbf{Question:} Does there exist a \CRA realizing $f$
    in the class $\mathcal{C}$ with at most $\sCRA$ states and at most $\rCRA$ registers?
  \end{itemize}
\end{definition}

In the sequel, we solve this problem for linear and affine \CRA:
\begin{restatable}{theorem}{thmStateRegMin}
\label{thm:state-reg-min}
The state-register minimization problem is decidable for the classes of linear and affine \CRA
in \nexptime.
Furthermore, the algorithm exhibits a solution when it exists.
\end{restatable}

\begin{remark}
  The complexities we give are valid for fields where it is possible 
  to perform elementary operations efficiently (\eg $\mathbb{Q}$).
  See Remark~\ref{rmk:cpxOnField} for a more detailed discussion on the matter.
\end{remark}

\section{Characterizing the state-register complexity using invariants of WA}
\label{sec:characterization}

\subsection{Zariski topologies and invariants of WA}
\label{subsec:Zariski}

Let $\mathbb{K}$ be a field and $n\in \mathbb{N}$.
The \emph{Zariski topology} on $\mathbb{K}^n$ is defined
as the topology whose closed sets are the sets of common roots of a finite collection of polynomials of $\mathbb{K}[X_1, \dots, X_n]$.
A linear version of this topology, called the \emph{linear Zariski topology},
was introduced by Bell and Smertnig in~\cite{BellS21}.
Its closed sets, which we will call \emph{Z-linear sets},
are finite unions of vector subspaces of $\mathbb{K}^n$.

A set $S \subseteq \mathbb{K}^n$ is called \emph{irreducible} if,
for all closed sets $C_1$ and $C_2$,
such that $S \subseteq C_1 \cup C_2$, we have either $S \subseteq C_1$ or $S \subseteq C_2$.
The Zariski topologies defined above are Noetherian topologies in which every closed set
can be written as a finite union of irreducible components.
We then define the \emph{dimension} of a Z-linear set as the maximum dimension of
its irreducible components and their number will be called its \emph{length}.

For a set $S \subseteq \mathbb{K}^n$, $\linClosure{S}$
will denote its closure in the linear Zariski topology.
In this topology, closed irreducible sets are vector subspaces of $\mathbb{K}^n$
and linear maps are continuous and closed maps (mapping closed sets to closed sets).
In particular, for all $S \subseteq \mathbb{K}^n$ and linear map
$f : \mathbb{K}^n \to \mathbb{K}^n$,$\linClosure{f(S)}= f(\linClosure{S})$.
Moreover, if $S \subseteq \mathbb{K}^n$ is irreducible and $f : \mathbb{K}^n \to \mathbb{K}^n$
is continuous, then $f(S)$ is irreducible.
These properties will be used implicitly in the following
(see~\cite[Lemma 3.5]{BellS21} for more details and references).

We will also define an affine version of this topology that enjoy the same properties
in Subsection~\ref{subsec:affine}.

\begin{definition}
  Let $\Sigma$ be a finite alphabet and let $\mathcal{R} = (u,\mu,v)$
  be a $\dWA$-dimensional \WA on $\Sigma$ over $\mathbb{K}$.
  A subset $I \subseteq \mathbb{K}^\dWA$ is called an \emph{invariant} of $\mathcal{R}$
  if $u \in I$ and, for all $w \in I$ and $a \in \Sigma$, $w \mu(a) \in I$.
  For two invariants $I_1$ and $I_2$, we say that $I_1$ is \emph{stronger} than $I_2$
  if $I_1 \subseteq I_2$.
  In particular, the strongest invariant of $\mathcal{R}$
  is its \emph{reachability set} $\lReachSet{\mathcal{R}} = u \mu(\Sigma^*)$.

  An invariant that is also a Z-linear set will be called a \emph{Z-linear} invariant.
  The strongest Z-linear invariant of $\mathcal{R}$ is the closure of
  $\lReachSet{\mathcal{R}}$ in the linear Zariski topology
  (which is well-defined since the topology is Noetherian).
\end{definition}

\begin{example}[Example~\ref{ex:WA} continued]
  \label{ex:linHullWA}
  The reachability set of the $\WA$ considered in Example~\ref{ex:WA}
  is $\lReachSet{\mathcal{R}} = \big\{ (2^{2n}, 0) \,\big|\, \allowbreak n \in \mathbb{N} \big\}
  \cup \left\{ (0, 2^{2n+1}) \,\middle|\, n \in \mathbb{N}\right\}$.
  Its \LH is then the union of the two coordinate axes of the plane
  $\linHull{\mathcal{R}} = \linSpan{1,0} \cup \linSpan{0,1}$.

  Indeed, the inclusion $\subseteq$ comes from the fact that $u = (1,0) \in \linSpan{1,0}$
  and $\linSpan{1,0} \cup \linSpan{0,1}$ is stable by multiplication by $\mu(a)$
  and the inclusion $\supseteq$ comes from the fact that, for the linear Zariski topology,
  $\left\{ (1,0) \right\}$ is dense in $\linSpan{1,0}$
  and $\left\{ (0,2) \right\}$ is dense in $\linSpan{0,1}$.
\end{example}

\begin{remark}\label{rk:lin_not_alg}
In the previous example, the strongest Z-linear invariant is actually the strongest
algebraic invariant (\emph{i.e.} closed in the Zariski topology).
Of course, this is not always the case.
\end{remark}

The Z-linear invariants of two \WA
realizing the same function do not necessarily coincide but,
since $\mathbb{K}$ is a field, it is well-known that for every rational series $f$,
there exists a (computable) minimal \WA realizing $f$
that is unique up to similarity in the following sense
(see~\cite[Proposition 4.10 (Chapter III)]{Sakarovitch09}):

\begin{definition}
  Let $\mathcal{R} = (u,\mu,v)$ and $\mathcal{R}' = (u',\mu',v')$
  be two $\dWA$-dimensional \WA over $\mathbb{K}$.

  $\mathcal{R}$ and $\mathcal{R}'$ are said to be \emph{similar}
  if there exists an invertible (change of basis) matrix $P \in \sqmatrSet{\mathbb{K}}{\dWA}$
  such that $u' = u P$, $\mu'(a) = P^{-1} \mu(a) P$ for all $a \in \Sigma$
  and $v' = P^{-1}v$.
\end{definition}

\begin{remark}
  \label{rmk:similarRep}
  The Z-linear invariants of two similar \WA $\mathcal{R}$ and $\mathcal{R}'$
  only differ by a change of basis.
  \ie there is a bijection between the Z-linear invariants
  of $\mathcal{R}$ and those of $\mathcal{R}'$ that, in particular, preserves the length and dimension.
\end{remark}

\subsection{Strongest invariants and  characterization}

The notion of \LH was introduced by Bell and Smertnig in~\cite{BellS21}, under the name ``linear hull''.
They showed, in~\cite{BS23}, that it is computable and can be used to decide
whether a \WA is equivalent to a deterministic (or an unambiguous) one.

\begin{theorem}[{\cite[Theorem 1.3]{BellS21}}]
  \label{thm:seqWA}
  A rational series $f$ can be realized by a sequential $\WA$
  iff the \LH of a minimal \WA realizing $f$ has dimension at most 1.
\end{theorem}

The following result 
generalizes this theorem by linking linear $\CRA$ to Z-linear invariants.
It constitutes the key characterization that will allow us to solve the minimization problems.
\begin{theorem}[Characterization]
  \label{thm:mainThm}
  Let $f$ be a rational series.
  Then $f$ can be realized by a linear $\CRA$ with $\sCRA$ states and $\rCRA$ registers
  iff there exists a minimal \WA realizing $f$ that has a Z-linear invariant
  of length at most $\lInv$ and dimension at most $\dInv$.
\end{theorem}

As we will see in Subsection~\ref{subsec:affine},
this theorem can also be extended to affine \CRA.

Observe that, thanks to Remark~\ref{rmk:similarRep}, the property of the above characterization
is actually valid for \emph{every} minimal \WA realizing $f$.
Moreover, since the dimension of the \LH is minimal, finding this dimension
allows to solve the register minimization problem for linear $\CRA$.
This is formalized in the following result, which generalizes
Theorem~\ref{thm:seqWA} thanks to Remark~\ref{rk:seq}.

\begin{corollary}
  \label{cor:minRegLH}
  The register complexity of a rational series $f$ w.r.t.\ the class of linear \CRA
  is the dimension of the \LH of any minimal \WA realizing $f$.
\end{corollary}

An immediate consequence of this result is that computing the strongest invariant
allows to decide the register minimization problem.

\begin{example}[Example~\ref{ex:WA} continued]
  As we have seen in Example~\ref{ex:linHullWA},
  $\linHull{\mathcal{R}}$ is 1-dimensional and has two irreducible components,
  thus $\llbracket \mathcal{R} \rrbracket$ can be realized by a
  $\CRA$ with two states and one register (depicted on the right of Figure~\ref{fig:CRA}).
\end{example}

\subsection{Invariants of minimal WA and correspondence with CRA}
\label{subsec:invWA_CRA}

\begin{proposition}
  \label{prop:dimLHMin}
  Let $\mathcal{R}$ be a \WA realizing a rational series $f$.
  If $\mathcal{R}$ has a Z-linear
  invariant of length $\lInv$ and dimension $\dInv$,
  then every minimal \WA realizing $f$ has a Z-linear
  invariant of length $\leq \lInv$ and dimension $\leq \dInv$.
\end{proposition}

Thanks to Remark~\ref{rmk:similarRep}, it suffices to show the existence of one
minimal \WA verifying the proposition, since they are all similar.
It is known (see Proposition~\ref{prop:caracMinRep}) that a minimal \WA can be obtained from a \WA by alternating
between two constructions which reduce the dimension to make it match
the one of the span of the left (resp. right) reachability set. The result then
follows from the next lemma, which states that both constructions
decrease the length and dimension of the invariants.
We prove it by considering an adequate change of basis, and
verifying that it preserves invariants.

\begin{lemma}
  \label{lem:leftRightMin}
  Let $\mathcal{R}$ be a \WA realizing a rational series $f$,
  let $S_{\mathcal{R}}$ be a Z-linear invariant of $\mathcal{R}$ of length $\lInv$ and dimension $\dInv$
  and let $r = \dim(\linSpan{\lReachSet{\mathcal{R}}})$.
  We can construct an $r$-dimensional \WA
  $\mathcal{R}'$ realizing $f$, with a Z-linear invariant $S_{\mathcal{R}'}$ of length $\leq \lInv$
  and dimension $\leq \dInv$.
  The same holds with $r = \dim(\linSpan{\rReachSet{\mathcal{R}}})$.
\end{lemma}

The next proposition allows to go from Z-linear invariants of \WA to \CRA.
This construction builds on the one of~\cite[Lemma 3.13]{BellS21}, in which they build an equivalent \WA from
the strongest Z-linear invariant of a \WA.
We show that an analogous construction is valid for any Z-linear invariant,
and that we can use states of \CRA to represent the different irreducible components of the invariant,
thus reducing the number of registers used to the dimension of the invariant.

\begin{proposition}
  \label{prop:linRepToCRA}
  Let $\mathcal{R}$ be a \WA.
  If $\mathcal{R}$ has a Z-linear invariant of length $\lInv$ and dimension $\dInv$,
  then there exists a linear \CRA $\mathcal{A}$, with $\sCRA$ states and $\rCRA$ registers,
  such that $\llbracket \mathcal{A} \rrbracket = \llbracket \mathcal{R} \rrbracket $.
\end{proposition}

The next proposition shows the converse direction, from \CRA to invariants
of \WA.
The construction is the classical one from \CRA to \WA.
The existence of the adequate invariant follows from the determinism of the \CRA which ensures
that in any reachable configuration, only coordinates associated with the reachable state of the \CRA can be non-zero.

\begin{proposition}
  \label{prop:CRATolinRep}
  Let $\mathcal{A}$ be a linear \CRA.
  If $\mathcal{A}$ has $\sCRA$ states and $\rCRA$ registers, then there exists a \WA $\mathcal{R}$,
  with a Z-linear invariant of length $\lInv$ and dimension $\dInv$,
  such that $\llbracket \mathcal{A} \rrbracket = \llbracket \mathcal{R} \rrbracket $.
\end{proposition}

Using the three previous propositions, we can finally prove the main characterization:
\begin{proof}[Proof of Theorem \ref{thm:mainThm}]
  Given a linear $\CRA$ with $\sCRA$ states and $\rCRA$ registers, we can construct,
  thanks to Proposition~\ref{prop:CRATolinRep},
  an equivalent \WA with a Z-linear invariant
  of length $\lInv$ and dimension $\dInv$.
  Then the desired minimal \WA exists thanks to Proposition~\ref{prop:dimLHMin}.

  Reciprocally, applying the construction of Proposition~\ref{prop:linRepToCRA}
  to any minimal \WA gives the desired linear $\CRA$.
\end{proof}

As we will discuss in the next subsection below,
the three propositions we used for this proof can also be adapted to yield the same result
for affine \CRA.

\subsection{Z-affine invariants and affine CRA} \label{subsec:affine}
All the results of Section~\ref{sec:characterization} can actually be extended to affine \CRA
using the \emph{affine Zariski topology} instead of the linear one.
It is a slight generalization of the linear Zariski topology where closed sets,
called \emph{Z-affine} sets, are finite unions of affine spaces instead of vector spaces,
with lengths and dimensions defined like in the linear case.
It is still a Noetherian topology coarser than the Zariski topology,
affine maps are continuous and closed maps in this topology
and, more broadly, it enjoys the same properties as the linear Zariski topology
we considered throughout this section.
For a set $S \subseteq \mathbb{K}^n$, we will denote by $\affClosure{S}$ it closure
in the affine Zariski topology and, similarly to the linear case,
for a \WA $\mathcal{R} = (u,\mu,v)$, we will call any invariant of $\mathcal{R}$ that is a Z-affine set
a \emph{Z-affine invariant} of $\mathcal{R}$.
Of course, the strongest Z-affine invariant of $\mathcal{R}$ is still the closure of its reachability set
\ie its ``affine hull'' $\affHull{\mathcal{R}}$ and Remark~\ref{rmk:similarRep} is still true
for Z-affine invariants.

We obtain the same characterization of Theorem~\ref{thm:mainThm} in the affine setting :
\begin{theorem}[Characterization]
  \label{thm:charAff}
  Let $f$ be a rational series.
  Then $f$ can be realized by an affine $\CRA$
  with $\sCRA$ states and $\rCRA$ registers iff there exists a minimal \WA realizing $f$ that has a Z-affine invariant
  of length at most $\lInv$ and dimension at most $\dInv$.
\end{theorem}

We can show that Propositions~\ref{prop:dimLHMin},~\ref{prop:linRepToCRA} and~\ref{prop:CRATolinRep}
are also true if we replace Z-linear invariants by Z-affine ones and linear \CRA by affine ones.
So, the proof of Theorem~\ref{thm:charAff} remains the same as Theorem~\ref{thm:mainThm}.
All the details can be found in Appendix~\ref{apx:proofChar}.

Of course, this theorem has the same consequences of its linear counterpart
and we obtain an affine version of Corollary~\ref{cor:minRegLH}
\begin{corollary}
  \label{cor:minRegLHAff}
  The register complexity of a rational series $f$ w.r.t.\ the class of affine \CRA
  is the dimension of the \AH of any minimal \WA realizing $f$.
\end{corollary}

Working in the affine Zariski topology instead of the linear one can decrease
the dimension of the strongest invariant by one, as shown in the following example.
\begin{example}
  \label{ex:sumPow2}
  On the alphabet $\Sigma = \left\{ a \right\}$,
  let $\mathcal{R} = (u, \mu, v)$,
  where $u = (1, 2)$, $\mu(a) = \begin{pmatrix}
                                  1 & 0 \\
                                  1 & 2
  \end{pmatrix}$ and $v = (1, 0)^t$,
  be a \WA (over $\mathbb{R}$) realizing the rational series $f$ defined by
  $f(a^n) = \sum_{i=0}^{n} 2^i = 2^{n+1}-1$.

  The reachability set of $\mathcal{R}$ is $\lReachSet{\mathcal{R}}
  = \big\{ \left(\sum_{i=0}^{n} 2^i , 2^{n+1}\right) \,\big|\, n \in \mathbb{N} \big\}$.

  For the linear Zariski topology, $\lReachSet{\mathcal{R}}$ is dense in $\mathbb{R}^2$.
  So the \LH $\linHull{\mathcal{R}} = \mathbb{R}^2$ is two-dimensional.
  However, note that, for all $(x,y) \in \lReachSet{\mathcal{R}}$, $y = x+1$.
  So, by an argument of density in the affine Zariski topology,
  the \AH $\affHull{\mathcal{R}}$ is the affine line $y=x+1$, which is one-dimensional.
\end{example}

Thus, in the case where the dimensions of the affine and linear hulls doesn't match,
using affine \CRA instead of linear \CRA can allow to save one register :

\begin{example}[Example~\ref{ex:sumPow2} continued]
  \label{ex:sumPow2CRA}
  The two $\CRA$ depicted on Figure~\ref{fig:sumPow2CRA} both realize the function of Example~\ref{ex:sumPow2}.
  On the left we have a linear $\CRA$ with two registers and, on the right,
  an affine $\CRA$ with only one register.
  The characterization theorems show that both have the minimal number of
  registers for their respective classes of $\CRA$.
\end{example}


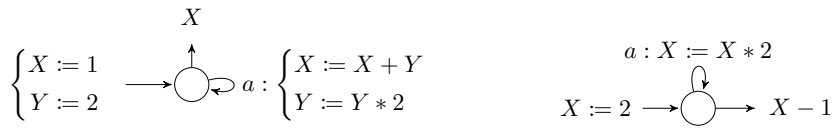
\begin{figure}[h!]
  \centering
  \vspace{-.5cm}
  \scalebox{.9}{
    \begin{tikzpicture}[->,>=stealth',shorten >=1pt,auto,node distance=1.8cm]
      \tikzstyle{every state}= [minimum size=5mm]

      \node[state] (p) at (0,0) {};

      \node[state,draw=none] (ps) at (-2,0) {$\left \{\begin{aligned}
                                                        X \coloneqq 1\\Y \coloneqq 2
      \end{aligned}\right. $};

      \node[state,draw=none] (pp) at (0,1) {$X$};

      \draw
      (p) edge [loop right] node {$a: \left \{\begin{aligned}
                                                & X \coloneqq X + Y \\& Y \coloneqq Y*2
      \end{aligned}\right. $} (p)
      (ps) edge (p)
      (p) edge (pp);
    \end{tikzpicture}
  }
  \hfil
  \scalebox{.9}{
    \begin{tikzpicture}[->,>=stealth',shorten >=1pt,auto,node distance=1.8cm]
      \tikzstyle{every state}= [minimum size=5mm]

      \node[state] (p) at (0,0) {};

      \node[state,draw=none] (ps) at (-1.5,0) {$X \coloneqq 2 $};

      \node[state,draw=none] (pp) at (1.5,0) {$X-1$};

      \draw
      (p) edge [loop above] node {$a : X \coloneqq X*2 $} (p)
      (ps) edge (p)
      (p) edge (pp);
    \end{tikzpicture}
  }
\vspace{-.3cm}
  \caption{Two CRA detailed in Example~\ref{ex:sumPow2CRA}.}
  \label{fig:sumPow2CRA}
\end{figure}

\section{Algorithms and complexity for the minimization problems}
\label{sec:algos}

We present two original algorithms to solve the minimization problems
we consider.
It is worth observing the difference between the two characterizations we have obtained: while the register complexity can be computed from a canonical object
(the strongest Z-linear invariant of the \WA), the state-register complexity is based on the existence
of a particular Z-linear invariant.
This explains why we derive a non-deterministic
procedure for the latter, and a deterministic for the former.

\subsection{Algorithm for the state-register minimization problem}
\label{subsec:sr}

We provide here a \nexptime algorithm for the state-register minimization problem,
hence proving Theorem~\ref{thm:state-reg-min}.
The algorithm runs in \nptime in $\sCRA$, $\rCRA$, and the size of the automaton.
The fact that $\sCRA$ is given in binary explains the exponential discrepancy.

\textbf{Small representations of Z-affine sets}
Let $\mathcal R=(u,\mu,v)$ be a \WA of dimension $\dWA$ over an alphabet $\Sigma$.
Let $L=A_1\cup\cdots\cup A_\lInv$ be a Z-affine set of length $\lInv$ of $\mathbb K^\dWA$.

An $\mathcal R$-representation $R$ of $L$ is a set of $\lInv$ finite sets of words $S_1,\ldots, S_\lInv$
such that
$\affSpan{\set{u\mu(w)|\ w\in S_i}}=A_i$ for all $i\in \set{1,\cdots,\lInv}$.
The \emph{size} of $R$ is the sum of the lengths of all words appearing in $R$.
The following key lemma shows that all Z-affine invariants of $\mathcal R$ have small
$\mathcal R$-representations, up to considering stronger invariants.

\begin{lemma}
  \label{lem:rep}
  Let $\mathcal R$ be a \WA.
  Let $I$ be a Z-affine invariant of $\mathcal R$ of
  length $\lInv$ and dimension $\dInv$.
  There exists an $\mathcal R$-representation $R$ 
  of size $\leq \lInv^2 \dInv^2$ of
  a Z-affine invariant $J\subseteq I$, of dimension $\leq \dInv$ and length $\leq \lInv$.
\end{lemma}

This property allows to derive the non-deterministic algorithm.
First, minimization of a \WA over a field can be performed in polynomial time
(see \eg~\cite[Corollary 4.17]{Sakarovitch09}).
Then, let $\mathcal R$ be a minimal \WA and let $\dInv,\lInv$ be positive integers.
  From Lemma~\ref{lem:rep}, we know that a Z-affine invariant of dimension $\dInv$ and length $\lInv$
  can be represented in size $O(\dInv^2 \lInv^2)$
  (up to finding a stronger invariant with smaller dimension and length).
  The algorithm works thusly: first step is to guess an $\mathcal R$-representation $R$ of a Z-affine set.
  The second step is to check that $R$ represents an invariant,
  which can be done easily using basic linear algebra.
  From this one can compute an affine \CRA with $\rCRA$ registers and $\sCRA$ states.
  Moreover, if we require that $R$ is Z-linear, we obtain a linear \CRA.
  If $R$ is not an invariant, the computation rejects.
  Note that different accepting computations may give rise to different invariants and thus different CRAs.

\subsection{Algorithm for the computation of Z-affine invariants}
\label{subsec:computation}

We describe a deterministic procedure which, given a \WA $\mathcal R$
and an integer $c$, returns a Z-affine invariant $J$ which is stronger
that any Z-affine invariant $I$ of $\mathcal R$ of length at most $c$.
When $c$ is chosen large enough, this procedure returns the strongest
Z-affine invariant of $\mathcal R$.
A similar procedure works as
well for the computation of Z-linear invariants.

\begin{wrapfigure}{R}{0.45\textwidth}
\begin{minipage}{0.45\textwidth}
   \vspace{-3ex}
\begin{algorithm}[H]
  \caption{Computing a Z-affine invariant}
  \label{algo:compute}
  \begin{algorithmic}[1]
    \Require{A \WA $\mathcal R = (u,\mu,v)$ of dimension $\dWA$, an integer $c$}
    \Ensure{A Z-affine invariant $J$ of $\mathcal R$ stronger than
    $I_c(\calR)$}
    \State $J \coloneqq \{u\}$
    \While{$J$ is not an invariant of $\mathcal R$}
      \State Pick some component $A$ of $J$, and some matrix $M$ of $\mathcal{R}$ s.t.
      $A\cdot M \not\subseteq J$
      \State $J \coloneqq J \cup A\cdot M$
      \If{$\mathrm{length}(J) > c^d$}
        \State $J \coloneqq \textsc{reduce}(J)$
      \EndIf
    \EndWhile
    \State \Return $J$
  \end{algorithmic}
\end{algorithm}
\end{minipage}
\end{wrapfigure}
Intuitively, this procedure will build a Z-affine set $J$ as follows:
it starts with a set containing only the initial vector of $\mathcal R$,
and incrementally extends it until it forms an invariant.
During this process, it should ensure that $J$ is included in every
Z-affine invariant $I$ of $\mathcal R$ of length at most $c$.
This relies on the following easy observation: if such an invariant $I$ contains at least
$c+1$ points on the same affine line (\emph{i.e.} a 1-dimensional affine space, denoted $D$),
then $I$ must have a component that contains $D$.
Indeed, as $I$ has length at most $c$, one of its components contains two such points.
As this component is irreducible, it is an affine subspace, hence contains $D$.
This reasoning can be lifted to higher dimensions as follows.

Given a \WA $\mathcal R$, and $c\in \mathbb N$, we denote by $I_c(\calR)=\bigcap_{\mathrm{length}(I)\leq c} I$
the intersection of all Z-affine invariants of $\mathcal{R}$ with at most $c$ components.

\begin{lemma}
  \label{lem:subspaces}
  Let $\mathcal R$ be a \WA and let $c,k\in \mathbb N$.
  Let $A_1,\ldots,A_{c^{k}+1}\subseteq I_c(\calR)$ be affine spaces such that:
  for any $P\subseteq \intInterv{1}{c^{k}+1}$ with $|P|\geq c^{k-1}+1$, $\affSpan{\cup_{i\in P}A_i}$ has
  dimension $k$.
  Then $\affSpan{\cup_{i\in \intInterv{1}{c^{k}+1}}A_i}\subseteq I_c(\calR)$.
\end{lemma}

Using this lemma, we derive an effective procedure
to simplify a Z-affine set $J=A_1\cup\cdots\cup A_{c^{\dWA}+1}$ by ``merging''
two components.
We denote by $\textsc{reduce}(J)$
the resulting set.
\begin{claim}
  \label{claim:reduce}
  Let $\mathcal R=(u,\mu,v)$ be a \WA of dimension $\dWA$, let $c\in \mathbb N$.
  Let $A_1,\ldots,A_{c^{\dWA}+1}\subseteq I_c(\calR)$ be affine spaces.
  One can find $1\leq i<j\leq c^\dWA+1$ such that $\affSpan{A_i\cup A_j}\subseteq I_c(\calR)$, in time $O(c^{p(d)})$, for some fixed polynomial $p$.
\end{claim}

\begin{theorem}
  \label{thm:cpxAlgoDet}
  Algorithm~\ref{algo:compute} is correct and terminates in time $O(c^{p(d)})$.
\end{theorem}

\begin{proof}
  Let us first discuss termination.
  Because of line $5$-$7$, the length of $J$ is at most $c^d+1$.
  Moreover $J$ is an increasing Z-affine set, thus its value can be modified at most $(d+1)\cdot (c^d+1)$ times,
  thus from Claim~\ref{claim:reduce} the algorithm terminates in time $O(c^{p(d)})$.

  We now discuss correctness.
  We need to show that $J$ is stronger than $I_c(\calR)$.
  Initially, this holds.
  Moreover, if $A\subseteq I_c(\calR)$ is an affine set, then for any $M\in \mu(\Sigma$)$, A\cdot M\subseteq I_c(\calR)$, since $I_c(\calR)$ is invariant.
  Thus, line $4$ preserves the property that $J$ is stronger than $I_c(\calR)$.
  Using Claim~\ref{claim:reduce}, the \textsc{Reduce} subroutine also preserves this property,
  since it only merges components whose affine span is contained in $I_c(\calR)$.
\end{proof}

\subsection{Complexity of the register minimization problem}
\label{subsec:reg}

In order to compute the strongest Z-linear and Z-affine invariants of a \WA
using Algorithms~\ref{algo:compute}, it is sufficient to be able to bound their lengths.
The following result gives such bounds.

\begin{theorem}\label{thm:lengths}
Let $\mathcal{R} = (u,\mu,v)$ be a $\dWA$-dimensional \WA on a finite alphabet $\Sigma$.
We have the following upper bounds :
\begin{itemize}
  \item The lengths of $\linHull{\mathcal{R}}$ and $\affHull{\mathcal{R}}$ are
  at most doubly-exponential in $\dWA$.
  \item If $\genMono{\mu(\Sigma)}$ is commutative (\emph{e.g.} $\Sigma$ is unary), then the length of $\linHull{\mathcal{R}}$
  is at most exponential in $\dWA$.
\end{itemize}
  We also have the following lower bound (which also hold for \WA over a unary alphabet):
  \begin{itemize}
    \item For all $\dWA > 0$, there exist a $\dWA$-dimensional \WA having
    strongest Z-linear and Z-affine invariants with lengths exponential in $\dWA$.
  \end{itemize}
\end{theorem}

\begin{proof}[Proof sketch]
The first item is shown in~\cite{BS23}, where the authors sketch a proof of a double-exponential upper bound
on the length of the \LH of a \WA, using tools from algebraic geometry,
which holds for $\mathbb{Q}$ in particular and for any field $\mathbb{K}$
where there is a double-exponential bound on the maximal order of finite groups of invertible matrices
(see~\cite[Proposition 48 and Remark 41]{BS23}). Their proof can be adapted
to $\affHull{\mathcal{R}}$.
The proof of the second item relies on basic linear algebra
and on results and ideas from~\cite{BS23} for invertible matrices
(see~\cite[Lemma 13 and Theorem 10]{BS23}).
Last, the lower bound is shown using a family of \WA
$(\mathcal{R}_i)_{i\in \mathbb{N}}$ whose dimension is polynomial in $i$
and \LH has a length that is exponential in $i$.
It is defined, using permutation matrices
of dimension $p$, for some prime number $p$, which generate 
cyclic groups.
The family is obtained by
using block matrices composed of such permutation matrices.
All the details are given in Appendix~\ref{app:length}.
\end{proof}

Thanks to this theorem, using Algorithm~\ref{algo:compute} with a large enough $c$
(at most doubly-exponential in the dimension of the given \WA),
and thanks to Theorem~\ref{thm:cpxAlgoDet}, we can prove the following result:
\begin{theorem}
  \label{thm:cpxLH}
  The \LAH of a \WA is computable in 2-\exptime.
\end{theorem}

This allows us to prove Theorem~\ref{thm:reg-min}.
Indeed, given a \WA $\calR$,
we first compute an equivalent minimal \WA, which can be done in polynomial time
(see \eg~\cite[Corollary 4.17]{Sakarovitch09}).
  Then, using Algorithm~\ref{algo:compute}, we compute the 
  \LrespAH of $\mathcal{R}$.
  Corollary~\ref{cor:minRegLH} (\resp Corollary~\ref{cor:minRegLHAff})
ensures that its dimension is the register complexity
  of $f$ w.r.t.\ the class of linear (\resp affine) \CRA, and
  the effectiveness follows from Proposition~\ref{prop:linRepToCRA}
(\resp its affine version).

Moreover, thanks to Theorem~\ref{thm:cpxLH} and the results of~\cite{BellS21},
we also have:
\begin{theorem}
  The sequentiality and unambiguity of a rational series are in 2-\exptime.
\end{theorem}

Note that the complexities of the last two theorems drop down to \exptime
when we  have a simply exponential bound
on the length of the strongest invariant.
This is the case when one considers
unary alphabets or \WA with commuting transition matrices in the linear setting,
as stated in Theorem~\ref{thm:lengths}.
In these cases, the bound is sharp.
It is still not clear however whether it is possible to close the gap between
the bounds in the general case.

\begin{remark}\label{rmk:cpxOnField}
  It is also worth noting that, while the characterizations that we obtained
  are valid for any field, the complexities of the algorithms are given in terms of number of
  elementary operations over the considered field.
  Which means that they hold for fields where we can perform basic operations in
  polynomial time (such as $\mathbb{Q}$ or its finite extensions).
  Moreover, the general upper bounds on the lengths given by Theorem~\ref{thm:lengths}
  were proven only for fields verifying a specific property (which is verified by $\mathbb{Q}$).
  See the proof for more details.
\end{remark}

\subsection{State/register tradeoff}
Reducing the number of registers may increase the number of states and vice-versa.
The following theorem summarizes what we know on this tradeoff.

\begin{theorem}\label{thm:tradeoff}
Let $f$ be a rational series realized by some $d$-dimensional \WA $\calR$.
Consider some pair of integers $(n,k)$ optimal for $f$ w.r.t.\ the class of linear \CRA.
The inequalities
$1\leq n \leq \mathrm{length}(\linHull{\mathcal{R}}) = O(2^{2^d})$
and $\dim(\linHull{\mathcal{R}}) \leq k \leq d$ hold true.

(They are valid in the affine setting as well)
\end{theorem}

\begin{remark}\label{rk:length}
Building the \CRA from the strongest invariant is not always optimal.
There are
some cases where it is possible to reduce
the number of states of a \CRA exponentially, while keeping the minimal number of registers,
by choosing an invariant that is weaker than the strongest Z-linear/Z-affine invariant but shorter.
\end{remark}

\section{Conclusion}

We have shown how to decide variants of \CRA minimization problems,
and have given complexity for the respective algorithms.
There are several ways in which these algorithms could be improved.
First, it would be worth reducing the gap between the lower and
the upper bounds on the length of the \LH.
Second, identifying a canonical invariant associated with the state-register
minimization problem would allow to derive a deterministic
algorithm for this problem.
Third, one could hope for better complexity if one only considers
the existence of equivalent CRA.
For instance, in~\cite{JeckerMP23} the authors give a \pspace algorithm
for the determinization problem (\ie $1$-register minimization problem)
in the case of a polynomially ambiguous automaton, via a quite different approach.

Another line of research consists in trying to use the techniques we developed
to solve the register minimization problem for other classes of \CRA,
for instance copyless \CRA
(which correspond to multi-sequential \WA).
Another ambitious goal is to consider register minimization in the context
of different semirings, but there all the linear algebra tools which are crucial
to solving these problems completely break down.
Similarly, it seems that register minimization for polynomial automata would be very difficult:
it was shown recently in~\cite{HrushovskiOPW23} that the strongest algebraic
invariant of a polynomial automaton is not computable.
One possibility may be to bound the ``degree'' of the invariants, where Z-affine sets
would correspond to algebraic sets of degree one.

\bibliography{biblio}

\clearpage
\appendix

\section{Complements to Section~\ref{sec:characterization}}

\subsection{Examples of Subsection~\ref{subsec:Zariski}}

\noindent
\textbf{Complement to Remark~\ref{rk:lin_not_alg}}

Consider a \WA
over a one-letter alphabet $\{a\}$ with $u = (1, 1)$ and $\mu(a) =
\begin{pmatrix}
  2 & 0 \\
  0 & 4
\end{pmatrix} $. One can verify that the strongest algebraic invariant is described by a polynomial
of degree $2$ (the second coordinate is the square of the first one),
hence is not linear.

\noindent
\textbf{Complement to Example~\ref{ex:sumPow2}}

A graphical representation of the reachability set of the WA described in
Example~\ref{ex:sumPow2} is depicted in Figure~\ref{fig:reachSet}.

All the reachable vectors are on the affine line $y = x+1$ so the affine hull is one dimensional
but, since the line does not cross the origin, the linear hull is two dimensional.


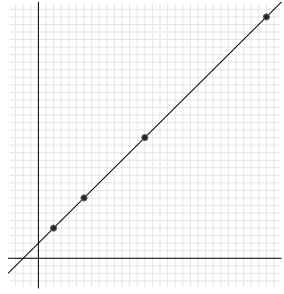
\begin{figure}[h!]
  \centering
  \scalebox{.2}{
    \begin{tikzpicture}
      \draw[step=0.5cm, gray!20,very thin] (-1.9,-1.9) grid (15.9,16.9);
      \draw [domain=-2:16, very thick] plot (\x, \x+1);
      \draw [domain=-2:16, thick] plot (\x, 0);
      \draw [domain=-2:17, thick] plot (0, \x);
      \filldraw[color=black!60, fill=black!85, very thick](1,2) circle (0.2);
      \filldraw[color=black!60, fill=black!85, very thick](3,4) circle (0.2);
      \filldraw[color=black!60, fill=black!85, very thick](7,8) circle (0.2);
      \filldraw[color=black!60, fill=black!85, very thick](15,16) circle (0.2);
    \end{tikzpicture}
  }
  \caption{Reachability set of the $\WA$ of Example~\ref{ex:sumPow2}.}
  \label{fig:reachSet}
\end{figure}

\subsection{Proofs of the building blocks of Theorem~\ref{thm:mainThm}} \label{apx:proofChar}
We will prove Lemma~\ref{lem:leftRightMin} (wich yields Proposition~\ref{prop:dimLHMin})
and Propositions~\ref{prop:linRepToCRA} and~\ref{prop:CRATolinRep} in the affine setting
(\ie using Z-affine sets and affine \CRA) as it is more general,
and all the results for the linear setting follow directly from the constructions of this section.

To simplify dealing with affine \CRA, we make the same observations we made in Section~\ref{sec:prelim}
about linear expressions (bringing back the underline notation to avoid confusion):
for a finite set of variables $\varSet=\left\{X_1, \dots, X_\rCRA \right\}$ that we can assume to be ordered,
we identify any affine expression
$e = \sum_{i=1}^{\rCRA} \alpha_i X_i + \beta$
with the affine form $\expToMap{e} \colon \mathbb{K}^\rCRA \to \mathbb{K}$
defined by $\expToMap{e}(u) = u \linPart{\expToMap{e}} + \affPart{\expToMap{e}}$
with $\linPart{\expToMap{e}} = (\alpha_1, \dots, \alpha_\rCRA)^t$ and $\affPart{\expToMap{e}} = \beta$.
We can then identify any affine substitution $s \colon \varSet \to \affExpr{\varSet}$
with the affine map $\expToMap{s} \colon \mathbb{K}^\rCRA \to \mathbb{K}^\rCRA$
defined by $\expToMap{s}(u) = u \linPart{\expToMap{s}} + \affPart{\expToMap{s}}$
with $\linPart{\expToMap{s}} = \linPart{(\expToMap{s(X_1)}} | \cdots| \linPart{\expToMap{s(X_\rCRA)}})$
and $\affPart{\expToMap{s}} = (\affPart{\expToMap{s(X_1)}}, \cdots, \affPart{\expToMap{s(X_\rCRA)}})$,
and we can identify any valuation $v \colon \varSet \to \mathbb{K}$ with the point
$\expToMap{v}=(v(X_1), \cdots, v(X_\rCRA))$ of the affine space $\mathbb{K}^\rCRA$.

We now drop the underline notation and observe that the registers of a linear (\resp affine)
\CRA and their updates can be characterized by the values of the vector (\resp point)
associated with $\noUnderline{v_0}$, and the linear (\resp affine) maps associated with the
$\noUnderline{\delta_{\mathcal{X}}(q,a)}$ and $\noUnderline{\outFctCRA(q)}$,
for all $q \in Q$ and $a \in \Sigma$, and we can check that
\[
  \llbracket \mathcal{A} \rrbracket (w) = \noUnderline{\outFctCRA(\delta_Q(q_0,w))}
  \left(\noUnderline{\delta_{\varSet}(q_0,w)} \left(\noUnderline{v_0}\,\right)\right)
\]

\subsubsection{Proof of Lemma~\ref{lem:leftRightMin}}

In the following, for all $\dWA \in \mathbb{N}$,
let $\canonBasis{\dWA} = \left\{ e_1, \dots, e_\dWA \right\}$
denote the canonical basis of $\mathbb{K}^\dWA$.
For two bases $B$ and $B'$ of the same vector space,
let $\chgBaseMatr{B}{B'}$ denote the change of basis matrix from $B$ to $B'$
(whose lines are the coordinates of the vectors of $B'$ in the basis $B$).
And finally, for two integers $i$ and $j$, let $I_i$ denote the identity matrix of size $i$
and let $\resizIdMatr{i}{j}$ denote the $i$ by $j$ matrix $(I_i \ |\ 0)$ if $i \leq j$
and $\resizIdMatr{j}{i}^t$ otherwise.

\begin{claim}
  \label{clm:chgBase}
  If $V$ is an $r$-dimensional vector subspace of $\mathbb{K}^\dWA$ and
  $B$ is a basis of $\mathbb{K}^\dWA$ whose first $r$ vectors form a basis of $V$,
  then, for all $v \in V$, since the $\dWA-r$ last entries of the vector
  $v \chgBaseMatr{\canonBasis{\dWA}}{B}$
  are all zeros, we note that
  $v \chgBaseMatr{\canonBasis{\dWA}}{B} \resizIdMatr{\dWA}{r} \resizIdMatr{r}{\dWA}
  \chgBaseMatr{B}{\canonBasis{\dWA}} = v$.
\end{claim}

\begin{proof}[Proof of Lemma \ref{lem:leftRightMin}]
  Let us first prove the lemma for the left reachability set.
  Let $\mathcal{R} = (u,\mu,v)$, let $\dWA$ be the dimension of $\mathcal{R}$ and
  let $B$ be a basis of $\vectSet{\mathbb{K}}{\dWA}$ obtained by completing a basis
  of $\linSpan{\lReachSet{\mathcal{R}}}$ with arbitrary vectors.

  We define $\mathcal{R}'$ as $(u',\mu',v')$ where, for all $a \in \Sigma$,
  \begin{gather*}
    u' = u \chgBaseMatr{\canonBasis{\dWA}}{B} \resizIdMatr{\dWA}{r} \quad
    v' = \resizIdMatr{r}{\dWA} \chgBaseMatr{B}{\canonBasis{\dWA}} v \quad
    \mu'(a) = \resizIdMatr{r}{\dWA} \chgBaseMatr{B}{\canonBasis{\dWA}} \mu(a)
    \chgBaseMatr{\canonBasis{\dWA}}{B} \resizIdMatr{\dWA}{r}
  \end{gather*}

  We can show by induction, thanks to Claim~\ref{clm:chgBase}, that
  $\lReachSet{\mathcal{R}'} = \lReachSet{\mathcal{R}}\chgBaseMatr{\canonBasis{\dWA}}{B} \resizIdMatr{\dWA}{
    r}$,
  and then $\llbracket \mathcal{R}' \rrbracket = \llbracket \mathcal{R} \rrbracket$.

  It remains to show that, if $\mathcal{R}$ has a Z-affine invariant $S_{\mathcal{R}}$
  of length $\lInv$ and dimension $\dInv$, then $\mathcal{R}'$ has a Z-affine invariant
  $S_{\mathcal{R}'}$ of length $\leq \lInv$ and dimension $\leq \dInv$.
  Note that we can assume that $S_{\mathcal{R}} \subseteq \linSpan{\lReachSet{\mathcal{R}}}$,
  since $S_{\mathcal{R}} \cap \linSpan{\lReachSet{\mathcal{R}}}$ is also a Z-affine
  invariant of $\mathcal{R}$, with a length $\leq \lInv$ and a dimension $\leq \dInv$.

  Let us take $S_{\mathcal{R}'} = S_{\mathcal{R}} \chgBaseMatr{\canonBasis{\dWA}}{B} \resizIdMatr{\dWA}{r}$
  and show that it has the right properties.
  First, it is clear that $S_{\mathcal{R}'}$ is a Z-affine invariant of $\mathcal{R}'$,
  since $u \in S_{\mathcal{R}}$ then $u' = u \chgBaseMatr{\canonBasis{\dWA}}{B} \resizIdMatr{\dWA}{r} \in
  S_{\mathcal{R}'}$
  and, for all $w' \in S_{\mathcal{R}}'$ and $a \in \Sigma$,
  there exists a $w \in S_{\mathcal{R}}$ such that
  $w' = w \chgBaseMatr{\canonBasis{\dWA}}{B} \resizIdMatr{\dWA}{r}$,
  thus, since $S_{\mathcal{R}} \subseteq \linSpan{\lReachSet{\mathcal{R}}}$
  and thanks to Claim~\ref{clm:chgBase}:
  \begin{align*}
    w' \mu'(a) & = w \chgBaseMatr{\canonBasis{\dWA}}{B} \resizIdMatr{\dWA}{r}
    \resizIdMatr{r}{\dWA} \chgBaseMatr{B}{\canonBasis{\dWA}} \mu(a)  \chgBaseMatr{\canonBasis{\dWA}}{B}
    \resizIdMatr{\dWA}{r} \\
    & = w \mu(a) \chgBaseMatr{\canonBasis{\dWA}}{B} \resizIdMatr{\dWA}{r} \in S_{\mathcal{R}}'
  \end{align*}
  and the image of a Z-affine set by a linear map is still a Z-affine set.
  Moreover, the irreducible components of $S_{\mathcal{R}'}$ are images of
  the irreducible components of $S_{\mathcal{R}}$ by a linear map.
  Thus $\dim(S_{\mathcal{R}'}) \leq \dim(S_{\mathcal{R}})$ and
  $\textup{\textsf{length}}(S_{\mathcal{R}'}) \leq \textup{\textsf{length}}(S_{\mathcal{R}})$.

  The case of the right reachability set is proven similarly,
  using a basis of $\linSpan{\rReachSet{\mathcal{R}}}$.
  Let $C$ be a basis of $\vectSet{\mathbb{K}}{\dWA}$ obtained by completing a basis
  of $\linSpan{\rReachSet{\mathcal{R}}}^t$ with arbitrary vectors.

  We define $\mathcal{R}''$ as $(u'',\mu'',v'')$ where, for all $a \in \Sigma$,
  \begin{gather*}
    u'' = u \chgBaseMatr{C}{\canonBasis{\dWA}}^t \resizIdMatr{\dWA}{r} \quad
    v'' = \resizIdMatr{r}{\dWA} \chgBaseMatr{\canonBasis{\dWA}}{C}^t v \quad
    \mu''(a) = \resizIdMatr{r}{\dWA} \chgBaseMatr{\canonBasis{\dWA}}{C}^t \mu(a)
    \chgBaseMatr{C}{\canonBasis{\dWA}}^t \resizIdMatr{\dWA}{r}
  \end{gather*}

  We can show by induction, thanks to Claim~\ref{clm:chgBase}, that
  $\rReachSet{\mathcal{R}'} = \resizIdMatr{r}{\dWA} \chgBaseMatr{\canonBasis{\dWA}}{C}^t
  \rReachSet{\mathcal{R}}$,
  and then $\llbracket \mathcal{R}' \rrbracket = \llbracket \mathcal{R} \rrbracket$.

  Like in the previous case, let $S_{\mathcal{R}}$ be a Z-affine invariant of $\mathcal{R}$,
  of length $\lInv$ and dimension $\dInv$ and let
  $S_{\mathcal{R}''} = S_{\mathcal{R}} \chgBaseMatr{C}{\canonBasis{\dWA}}^t \resizIdMatr{\dWA}{r}$.
  Obviously, $u'' \in S_{\mathcal{R}''}$ and,
  if $w'' = w \chgBaseMatr{C}{\canonBasis{\dWA}}^t \resizIdMatr{\dWA}{r} \in S_{\mathcal{R}''}$,
  then we can show that, for all $a \in \Sigma$, $w'' \mu''(a) \in S_{\mathcal{R}''}$
  by observing that $\chgBaseMatr{\canonBasis{\dWA}}{C}^t \mu(a)
  \chgBaseMatr{C}{\canonBasis{\dWA}}^t \resizIdMatr{\dWA}{r}$
  is a $\dWA$ by $r$ matrix whose last $\dWA-r$ rows are all null.
  Thus $ \resizIdMatr{\dWA}{r} \resizIdMatr{r}{\dWA} \chgBaseMatr{\canonBasis{\dWA}}{C}^t \mu(a)
  \chgBaseMatr{C}{\canonBasis{\dWA}}^t \resizIdMatr{\dWA}{r}
  = \chgBaseMatr{\canonBasis{\dWA}}{C}^t \mu(a) \chgBaseMatr{C}{\canonBasis{\dWA}}^t \resizIdMatr{\dWA}{r}$.


  Note that the construction only uses linear maps and, thus, still holds for Z-linear invariants.
\end{proof}

\subsubsection{Proof of Proposition~\ref{prop:linRepToCRA}}

\begin{proof}
  Let $\mathcal{R} = (u,\mu,v)$, let $\dWA$ be the dimension of $\mathcal{R}$
  and let $W_1, \dots, W_\lInv$ be the irreducible components of a Z-affine invariant of $\mathcal{R}$
  of length $\lInv$ and dimension $\dInv$.
  We assume, without loss of generality, that $u \in W_1$.

  For all $i \in \intInterv{1}{\lInv}$, let $W_i = p_i + V_i$ with $p_i \in \mathbb{K}^\dWA$
  and $V_i$ a vector subspace of $\mathbb{K}^\dWA$ and let $B_i$ be a basis of $\mathbb{K}^\dWA$
  obtained by completing a basis of $V_i$ with arbitrary vectors.

  We define $\mathcal{A}$ as $(Q, q_0, \varSet, v_0, \outFctCRA, \delta)$ where:
  \begin{itemize}
    \item $Q = \intInterv{1}{\sCRA}$ and $q_0 = 1$.
    \item $\varSet = \left\{ X_1, \dots, X_\rCRA \right\}$ and
    $\noUnderline{v_0} = (u-p_{q_0}) \chgBaseMatr{\canonBasis{\dWA}}{B_{q_0}} \resizIdMatr{\dWA}{\rCRA}$
    \item for all $q \in Q$, $x \in \mathbb{K}^\rCRA$,
    $\noUnderline{\outFctCRA(q)}(x) = \left( p_q + x
    \resizIdMatr{\rCRA}{\dWA} \chgBaseMatr{B_q}{E_\dWA} \right)v$.
    \item for all $q \in Q$ and $a \in \Sigma$, let $q' \in
    \left\{ p \in \llbracket 1,\sCRA \rrbracket \,\middle|\, W_q \mu(a) \subseteq W_p \right\}$
    (chosen arbitrarily).
    $\delta(q,a)$ will be defined by $\delta_Q(q,a) = q'$ and, for all $x \in \mathbb{K}^\rCRA$,
    \[\noUnderline{\delta_{\mathcal{X}}(q,a)}(x) =
    \left( \left( p_q + x \resizIdMatr{\rCRA}{\dWA}\chgBaseMatr{B_q}{E_\dWA} \right) \mu(a) - p_{q'}\right)
    \chgBaseMatr{E_\dWA}{B_{q'}} \resizIdMatr{\dWA}{\rCRA}\]
  \end{itemize}

  We can show by induction, using Claim~\ref{clm:chgBase}, that, for all $w \in \Sigma^*$,
  \[\noUnderline{\delta_{\mathcal{X}}(q_0,w)} \left( \noUnderline{v_0} \right) =
  \left( u \mu(w) - p_{\delta_Q(q_0,w)} \right)
  \chgBaseMatr{E_\dWA}{B_{\delta_Q(q_0,w)}} \resizIdMatr{\dWA}{\rCRA}\]

  Thus, for all $w \in \Sigma^*$,
  \[\llbracket \mathcal{A} \rrbracket(w) = \noUnderline{\outFctCRA(\delta_Q(q_0,w))}
  \left(\noUnderline{\delta_{\varSet}(q_0,w)} \left(\noUnderline{v_0}\,\right)\right) \allowbreak
  = \llbracket \mathcal{R} \rrbracket (w)\]

  The same construction can be applied to obtain the linear version of the proposition
  (it is the case where $p_i = 0$ for all $i \in \intInterv{1}{n}$).
\end{proof}

\subsubsection{Proof of Proposition~\ref{prop:CRATolinRep}}

\begin{proof}
  Let's assume, without loss of generality, that
  $\mathcal{A} = (Q, q_0, \allowbreak \varSet,  v_0, \outFctCRA, \delta)$
  is an affine \CRA where $Q = \intInterv{1}{\sCRA}$, $q_0 = 1$,
  $\mathcal{X} = \left\{ X_1, \dots, X_\rCRA \right\}$.

  We define $\mathcal{R}$ as $(u,\mu,v)$, where:
  \begin{itemize}
    \item $u = (\noUnderline{v_0}, 1 , 0, \dots, 0) \in \vectSet{\mathbb{K}}{\sCRA(\rCRA+1)}$
    \item $v = \left(\begin{array}{c}
                       \outFctCRA_1 \\
                       \vdots       \\
                       \outFctCRA_n
    \end{array}\right) \in \matrSet{\mathbb{K}}{\sCRA(\rCRA+1)}{1}$.
    where, for all $i \in \intInterv{1}{\sCRA}$,
    $\outFctCRA_i = \left(\begin{array}{c}
                            \linPart{\noUnderline{\outFctCRA(i)}} \\
                            \affPart{\noUnderline{\outFctCRA(i)}}
    \end{array}\right)$.
    \item for all $a \in \Sigma$, $\mu (a) =
    \left( \begin{array}{c|c|c}
             \delta_{1,1}(a)     & \dots  & \delta_{1,\sCRA}(a)     \\
             \hline
             \vdots              & \ddots & \vdots                  \\
             \hline
             \delta_{\sCRA,1}(a) & \dots  & \delta_{\sCRA,\sCRA}(a)
    \end{array} \right) $
    where, for all $i,j \in \intInterv{1}{\sCRA}$, $\delta_{i,j}(a) = \left(
    \begin{array}{c|c}
      \linPart{\noUnderline{\delta_{\mathcal{X}}(i,a)}} & 0 \\[5pt]
      \hline
      \affPart{\noUnderline{\delta_{\mathcal{X}}(i,a)}} & 1
    \end{array} \right)\in \sqmatrSet{\mathbb{K}}{\rCRA+1}$ if $\delta_{Q}(i,a) = j$ and 0 otherwise.
  \end{itemize}

  We can show by induction that the definition of the $\delta_{i,j}$
  extends to words and thus, for all $w \in \Sigma^*$,
  \begin{align*}
    \llbracket \mathcal{R} \rrbracket (w)
    &= \left( \noUnderline{v_0}\ \linPart{\noUnderline{\delta_{\mathcal{X}}(i,w)}}
    + \affPart{\noUnderline{\delta_{\mathcal{X}}(i,w)}} \right)
    \linPart{\noUnderline{\outFctCRA(\delta_Q(q_0,w))}}
    + \affPart{\noUnderline{\outFctCRA(\delta_Q(q_0,w))}}\\
    &= \noUnderline{\outFctCRA(\delta_Q(q_0,w))}
    \left(\noUnderline{\delta_{\varSet}(q_0,w)} \left(\noUnderline{v_0}\,\right)\right)\\
    \llbracket \mathcal{R} \rrbracket (w) &= \llbracket \mathcal{A} \rrbracket (w)
  \end{align*}

  Let $S = \bigcup_{i=1}^{\sCRA} \psi_i(\mathbb{K}^\rCRA)$ where, for all $i \in \intInterv{1}{\sCRA}$,
  $\psi_i : \mathbb{K}^\rCRA \to \mathbb{K}^{\sCRA(\rCRA+1)}$
  maps every vector $v \in \mathbb{K}^\rCRA$ to the vector of $\mathbb{K}^{\sCRA(\rCRA+1)}$
  that has $(v,1)$ as its $i$-th ``block'' of size $\rCRA+1$ and zeros everywhere else.

  We show that $S$ is the desired Z-affine invariant of $\mathcal{R}$ :
  First, $S$ is an invariant of $\mathcal{R}$, since
  $u = \psi_1(\noUnderline{v_0}) \in S$ and
  $\psi_i(x) \mu(a) = \psi_{\delta_Q (i,a)}(\delta_{\varSet}(i,a)(x)) \in S$,
  for all $\psi_i(x) \in S$ and $a \in \Sigma$.
  Then, since the $\psi_i$ are affine maps, $S$ is a Z-affine set,
  whose irreducible components are the $\psi_i(\mathbb{K}^\rCRA)$, which have a dimension
  $\dInv$ as required.

  Observe that the proposition is also proved in the linear setting,
  since applying this construction to a linear \CRA yields a \WA with the desired Z-linear invariant.
\end{proof}

\section{Complements to Section~\ref{sec:algos}}

\subsection{Proof of Lemma~\ref{lem:rep}}

\begin{proof}
  Let $\mathcal R=(u,\mu,v)$ be a \WA over alphabet $\Sigma$.
  Let $I$ be a Z-affine invariant of $\mathcal R$ of length $\lInv$ and dimension $\dInv$, with irreducible
  components $A_1,\ldots, A_\lInv$.

  We will define sets $S_1,\ldots, S_\lInv$ such that, for all $i\in \set{1,\cdots,\lInv}$,
  $\affSpan{\set{u\mu(w)|\ w\in S_i}}\subseteq A_i$, representing a Z-affine set $J$.
  Given a word $w$, we will write $u_w=u\mu(w)$.
  Since $I$ is an invariant, there is an index $i$ such that $u_\epsilon=u\in A_i$.
  Thus, the set $S_i$ is initialized as $\set{\epsilon}$ and all other sets are initially empty.

  If $J$ is not an invariant, this means there exist $i\in \set{1,\ldots, \lInv}$,  $w\in S_i$ and
  $a\in\Sigma$ such that $u_{wa}\notin J$.
  However, we know that since $I$ is an invariant, it must at least contain the reachability set of
  $\mathcal R$ and thus $u_{wa}\in I$.
  Hence, there exists $j$ such that $u_{wa}\in A_j$.
  Thus, we update $S_j\coloneqq S_j\cup \set{w}$.
  Note that we still have that $\set{u_s|\ s\in S_j}\subseteq A_j$ thus the invariant\footnote{We heard you
  liked invariants, so we put invariants in your invariants.} $J\subseteq I$ is preserved.

  A first observation is that each time we add a new vector to the representation of $J$, the dimension of
  one of the components increases by $1$, hence this can happen at most $(\dInv+1)\lInv-1$ times (assuming the
  dimension of the empty set is $-1$). Similarly, each time this happens, the maximum length of a word in
  the representation increases by at most $1$.

  Thus, after at most $(\dInv+1)\lInv-1$ steps, $J$ is an invariant included in $I$, of length $\leq n$
  , with a representation
  of size
  at most $O(\dInv^2 \lInv^2)$.
\end{proof}

\subsection{Proof of Lemma~\ref{lem:subspaces}}

\begin{proof}
  Let $\mathcal R=(u,\mu,v)$ be a \WA, let $c,k\in \mathbb N$ and
  let $A_1,\ldots, \allowbreak A_{c^{k}+1}\subseteq I_c(\calR)$ be spaces satisfying the assumptions.

  Assume by contradiction that
  $\affSpan{\cup_{i\in \intInterv{1}{c^{k}+1}}A_i}\not\subseteq I_c(\calR)$.
  Then, by definition, there must exist an invariant $I$ of length $\leq c$
  such that $\affSpan{\cup_{i\in \intInterv{1}{c^{k}+1}}A_i}\not\subseteq I$.
  This means that any component of $I$ can contain at most $c^{k-1}$ $A_i$s
  (or else it would contain a $k$-dimensional component included in
  $\affSpan{\cup_{i\in \intInterv{1}{c^{k}+1}}A_i}$, which is $k$-dimensional as well).
  As a consequence, $I$ must have at least $c+1$ components, which yields a contradiction.
\end{proof}

\subsection{Proof of Claim~\ref{claim:reduce}}

\begin{proof}
  Let $\mathcal R=(u,\mu,v)$ be a \WA, of dimension $\dWA$, let $c\in \mathbb N$.
  We show by induction on $k$ that if $A_1,\ldots,A_{c^{k}+1}\subseteq I_c(\calR)$ generates a subspace of
  dimension $\leq k$ then there exist $1\leq i<j\leq c^k+1$ such that
  $\affSpan{A_i\cup A_j}\subseteq I_c(\calR)$, and these indices can be found in time $O(c^{p(d)})$.

  For $k=0$, if two subspaces generate an affine subspace of dimension $0$, it means that $A_1=A_2$, which
  are singletons.
  Thus $\affSpan{A_1\cup A_2}=A_1\subseteq I_c(\calR)$.
  Assume the result holds for some $k$, let us show it holds for $k+1$.
  Let $A_1,\ldots,A_{c^{k+1}+1}\subseteq I_c(\calR)$ be subspaces generating a space
  of dimension $\leq k+1$.
  We want to find, if it exists, a set $P\subseteq \intInterv{1}{c^{k+1}+1}$ with $|P|= c^{k}+1$, such that
  $\affSpan{\cup_{i\in P}A_i}$ has
  dimension $\leq k$.
  If it exists then we use the induction assumption.
  Otherwise, Lemma~\ref{lem:subspaces} yields the result.
  We have to show that we can find such a set $P$, if it exists, in the given time constraint.
  Let $Q\subseteq \intInterv{1}{c^{k+1}+1}$, we denote by $A_Q=\affSpan{\cup_{i\in Q}A_i}$ the affine space
  associated with $Q$.
  We say that $Q$ is \emph{minimal} if it is minimal, inclusion-wise, among the sets associated with $A_Q$.
  Note that a minimal set has cardinality at most $d$.
  Thus there are at most $(c^d +1)^{d+1}\leq c^{d^2+2d+1}$ such sets.
  We say that $i\leq c^{k+1}+1$ is \emph{compatible} with $Q$ if $A_i\subseteq A_Q$.
  Given a set $Q$ one can obtain the set of indices compatible with $Q$ in time $c^{k+1}$.
  To find a set $P$ as above, one simply has to enumerate the minimal sets for strict subsets of
  $A_{\intInterv{1}{c^{k+1}+1}}$,
  and check if one of them has a set of compatible indices of size $\geq c^k +1$.
\end{proof}

\subsection{Proof of Theorem~\ref{thm:lengths}}
\label{app:length}

The first item of the statement has already been explained in the core of the paper.

We consider the proof of the simply exponential upper bound in the case where the transition matrices
of the given \WA commute (and in particular, for any \WA on a unary alphabet).
We will show it, using results and ideas from~\cite{BS23} and some basic linear algebra:

\begin{theorem}
  \label{thm:casCommut}
  If $\mathcal{R} = (u,\mu,v)$ is a $\dWA$-dimensional \WA on a finite alphabet $\Sigma$
  such that $\genMono{\mu(\Sigma)}$ is commutative, then the length of $\linHull{\mathcal{R}}$
  is, at most, exponential in $\dWA$.

  Thus, in this case, both the computation of the \LH and
  the resolution of the register minimization problem for linear \CRA can be done in \exptime.
\end{theorem}

The proof of this theorem relies on the two following lemmas:
The first provides a constant we will use to bound the length of the \LH
and characterizes it for minimal \WA on a unary alphabet
(see~\cite[Lemma 13 and Theorem 10]{BS23}):
\begin{lemma}
  \label{lem:powN}
  For all $\dWA \in \mathbb{N}$, there exists an integer $N \in \mathbb{N}$
  such that, for all invertible matrix $A \in \sqmatrSet{\mathbb{K}}{d}$,
  \begin{enumerate}
    \item for all vector subspace $V \subseteq \mathbb{K}^\dWA$,
    if $V A^n = V$ for some $n > 0$ then $V A^N = V$.
    \item \(\linClosure{\genMono{A}} = \bigcup_{i=0}^{N-1} A^i \linSpan{\genMono{A^N}}\).
  \end{enumerate}
\end{lemma}

Note that the value of $N$ depends only on $\mathbb{K}$ and $\dWA$ and is exponential in $\dWA$.

Indeed, in the proof of~\cite[Theorem 10]{BS23} (see also~\cite[Lemma 12]{BS23})
$N$ is taken to be the least common multiple of all $n \in \mathbb{N}$
such that $\phi(n) \leq [\mathbb{K}:\mathbb{Q}] \dWA$,
where $\phi$ is Euler's totient function and $[\mathbb{K}:\mathbb{Q}]$
is the degree of the field extension and a well-known property of $\phi$ is that,
for all $n \geq 1$, $\sqrt {\frac{n}{2}} \leq \phi(n) \leq n-1$.
Thus $N \leq \left( 2 [\mathbb{K}:\mathbb{Q}]^2 \dWA^2\right)!$

The second lemma is a direct consequence of~\cite[Lemma 9]{BS23}:
\begin{lemma}
  \label{lem:uniqueCompoId}
  If $S \subseteq \sqmatrSet{\mathbb{K}}{\dWA}$ is a monoid,
  then $\linClosure{S}$ has a unique irreducible component that contains the identity matrix $I_{\dWA}$.
\end{lemma}

Last, we focus on the third statement of Theorem~\ref{thm:lengths}, showing that
the exponential bound is sharp. This is shown by the following example,
where we define a sequence of \WA $(\mathcal{R}_i)_{i\in \mathbb{N}}$
with a dimension that is polynomial in $i$ and a \LH with a length that is exponential in $i$.
\begin{example}
  For all $i \in \mathbb{N}$, let $\mathbb{P}_{\leq i}$ denote the set of prime numbers up to $i$,
  let $M_p =
  \left( \begin{array}{ccccc}
           0      & 1 & 0      & \dots & 0      \\
           0      & 0 & 1      & \dots & 0      \\
           \vdots &   & \ddots &       & \vdots \\
           0      & 0 & 0      & \dots & 1      \\
           1      & 0 & 0      & \dots & 0
  \end{array} \right) \in \sqmatrSet{\mathbb{K}}{p}$,
  and for all $p \in \mathbb{P}_{\leq i}$
  (it is the permutation matrix associated with the $p$-cycle $(1\, 2\, \dots\, p)$).

  Let $d = \sum_{p \in \mathbb{P}_{\leq i}} p$ and $\mathcal{R}_i = (u_i,\mu_i,v_i)$
  be a $d$-dimensional \WA, on the unary alphabet $\Sigma = \left\{ a \right\} $,
  where $u_i = v_i^t = (1\, 2\, \dots\, d)$ and $\mu_i(a) = \left( \begin{array}{c|c|c|c}
                                                                     M_2    & 0     & \dots  & 0
                                                                     \\
                                                                     \hline
                                                                     0      & M_3   & \dots  & 0
                                                                     \\
                                                                     \hline
                                                                     \vdots &       & \ddots & \vdots
                                                                     \\
                                                                     \hline
                                                                     0      & \dots & 0      & M_
                                                                       {\max(\mathbb{P}_{\leq i})}
  \end{array} \right)$ is a $d$-dimensional block diagonal matrix.

  Observe that the monoid generated by $\mu_i(a)$ is a cyclic group of order
  $n =\textup{\textsf{lcm}}(\mathbb{P}_{\leq i}) = \prod_{p \in \mathbb{P}_{\leq i}} p$.
  Thus, the \LH of $\mathcal{R}_i$ has a length of $n$, since it is the union of $n$ lines.

\end{example}

\subsection{Proof of Theorem~\ref{thm:casCommut}} \label{apx:casCommut}

First, observe that, for all matrices $M \in \sqmatrSet{\mathbb{K}}{\dWA}$,
$(\ker{M^i})_{i\in \mathbb{N}}$ is an increasing sequence of vector subspaces of $\mathbb{K}^\dWA$.
Thus, for all $i \in \mathbb{N}$, $\ker{M^\dWA} = \ker {M^{\dWA+i}}$ and
$\mathbb{K}^\dWA = \im{M^\dWA} \oplus \ker{M^\dWA}$.

Let $r$ be the rank of $M^\dWA$ and $B$ be a basis of $\mathbb{K}^\dWA$
whose first $r$ (\resp last $\dWA-r$) vectors form a basis of $\im{M^\dWA}$ (\resp $\ker{M^\dWA}$).
Then, $\chgBaseMatr{B}{\canonBasis{\dWA}} M^\dWA \chgBaseMatr{\canonBasis{\dWA}}{B}$
is a matrix of the form $\left(
\begin{array}{c|c}
  M' & 0 \\
  \hline
  0  & 0
\end{array} \right)$ with $M' \in \sqmatrSet{\mathbb{K}}{r}$ an invertible matrix.

We can then bound the length of
$\linClosure{\genMono{M}} = \bigcup_{i=0}^{\dWA-1} M^i \linClosure{\genMono{M^\dWA}}$
using the bound for a single invertible matrix of Lemma~\ref{lem:powN}
by observing that $\genMono{M^\dWA} = \chgBaseMatr{\canonBasis{\dWA}}{B}
\genMono{\chgBaseMatr{B}{\canonBasis{\dWA}} M^\dWA \chgBaseMatr{\canonBasis{\dWA}}{B}}
\chgBaseMatr{B}{\canonBasis{\dWA}}$,
and $\genMono{M'}$ is isomorphic (as a semigroup) to
$\genMono{\chgBaseMatr{B}{\canonBasis{\dWA}} M^\dWA \chgBaseMatr{\canonBasis{\dWA}}{B}}$
by an isomorphism which is also a linear map.
Thus, $\textup{\textsf{length}}(\linClosure{\genMono{M^\dWA}}) \leq
\textup{\textsf{length}}(\linClosure{\genMono{M'}})$
which is exponential in $r \leq \dWA$.

This approach can be generalized to an arbitrary number of matrices $M_1, \dots, M_n$,
provided that $\genMono{M_1, \dots, M_n}$ is commutative, as we show next.

If  $M_1, \dots, M_n$ are invertible $\linClosure{\genMono{M_1, \dots, M_n}}$
has a unique irreducible component $W$ containing the identity matrix
(thanks to Lemma~\ref{lem:uniqueCompoId}) and, for all $i \in \intInterv{1}{n}$,
$M_i$ acts by permutation on a finite number of irreducible components of
$\linClosure{\genMono{M_1, \dots, M_n}}$, thus there exists $N_i > 0$
such that $W M_i^{N_i} = W$.
We can take $N_1 = \dots = N_n = N$ given by Lemma~\ref{lem:powN}
(which is exponential in $\dWA$).
Thus, the length of $\linClosure{\genMono{M_1, \dots, M_n}} =
\bigcup_{j_i \in \intInterv{0}{N-1}} M_1^{j_1}M_2^{j_2} \cdots M_n^{j_r} W $
is at most $N^n$.
Like previously, we can use this result to bound the length in the general case.

For all $I \subseteq \intInterv{1}{n}$, let $M_I = \left\{ M_i \,\middle|\, i \in I \right\}$
and let, for all $k \in \mathbb{N}$,
\[\genMono{M_I}^{\leq k} =
\left\{ \prod_{i \in I} M_i^{j_i} \,\middle|\, \forall i \in I, j_i \in \intInterv{0}{k}  \right\}\]
then
\[
  \genMono{M_1, \dots, M_n} = \left(\left(\prod_{i=1}^n M_i \right)^\dWA
  \genMono{M_1, \dots, M_n} \right)
  \bigcup
  \left( \bigcup_{k=1}^{n-1} \left(
  \bigcup_{\substack{I \subseteq \intInterv{1}{n}\\|I| = k}}
  \genMono{M_I}^{\leq \dWA-1} \genMono{M_{\intInterv{1}{n}\setminus I}} \right)\right)
\]
Thus, it suffices to show that
$\linClosure{\left(\prod_{i=1}^n M_i \right)^\dWA \genMono{M_1, \dots, M_n}}$
is of length at most exponential in $\dWA$.

Let $r$ be the rank of $\left( \prod_{i=1}^n M_i \right)^\dWA$ and let $B$ be a basis of $\mathbb{K}^\dWA$
whose first $r$ (\resp last $\dWA-r$) vectors form a basis of
$\im{\left( \prod_{i=1}^n M_i \right)^\dWA}$ (\resp $\ker{\left(\prod_{i=1}^n M_i \right)^\dWA}$).

Then, $\chgBaseMatr{B}{\canonBasis{\dWA}} \left(\prod_{i=1}^n M_i \right)^\dWA \chgBaseMatr{\canonBasis{
  \dWA}}{B} = \left(
\begin{array}{c|c}
  M' & 0 \\
  \hline
  0  & 0
\end{array} \right)$ for some invertible matrix $M' \in \sqmatrSet{\mathbb{K}}{r}$.

For all $i \in \intInterv{1}{n}$, $\im{\left( \prod_{j=0}^n M_j \right)^\dWA} =
\im{\left( \prod_{j=0}^n M_j \right)^\dWA M_i}$.
Thus, $\chgBaseMatr{B}{\canonBasis{\dWA}} M_i \chgBaseMatr{\canonBasis{\dWA}}{B}$
has the form $\left(
\begin{array}{c|c}
  M'_i  & 0      \\
  \hline
  M''_i & M'''_i
\end{array} \right)$
with $M'_i \in \sqmatrSet{\mathbb{K}}{r}$ invertible
because $M'$ is invertible and $\min_{i \in \intInterv{1}{n}}\textsf{rank}(M_i') \geq
\textsf{rank}\left(\left(\prod_{i=1}^n M_i'\right)^\dWA\right)
= \textsf{rank}(M') =r$.

Since
\[
  \linClosure{\left(\prod_{i=1}^n M_i \right)^\dWA \genMono{M_1, \dots, M_n}}
  =
  \chgBaseMatr{\canonBasis{\dWA}}{B}\linClosure{\chgBaseMatr{B}{\canonBasis{\dWA}}
  \left(\prod_{i=1}^n M_i \right)^\dWA \genMono{M_1, \dots, M_n}
  \chgBaseMatr{\canonBasis{\dWA}}{B}}\chgBaseMatr{B}{\canonBasis{\dWA}}
\]

we conclude by noting that  $M' \genMono{M_1', \dots, M_n'}$
is isomorphic (as a semigroup)
to\\  $\chgBaseMatr{B}{\canonBasis{\dWA}}
\left(\prod_{i=1}^n M_i \right)^\dWA \genMono{M_1, \dots, M_n}
\chgBaseMatr{\canonBasis{\dWA}}{B}$ by an isomorphism which is a linear map.

\subsection{Proof of Theorem~\ref{thm:tradeoff}}

The inequalities stated in Theorem~\ref{thm:tradeoff}
are graphically depicted on Figure~\ref{fig:tradeoff}.


\begin{figure}[h!]
  \centering
  \scalebox{.5}{
    \begin{tikzpicture}
	\draw[help lines, step=.5cm, color=gray!90, dashed] (0.5,-1.9) grid (14.9,4.9);
      	\draw[->,ultra thick] (-.5,-2)--(15,-2) node[right]{nb of states};
	\draw[->,ultra thick] (0.5,-3)--(0.5,5) node[above]{nb of registers};

      \draw [domain=1:14, very thick, variable=\x] plot ({\x}, {1/\x*4});

	\draw[very thick] (1,-1.8) -- (1,-2.2) node[below]{$1$};
	\draw[very thick] (14,-1.8) -- (14,-2.2) node[below]{$n'\le \textsf{length}(\linHull{\mathcal{R}})$};

	\draw[very thick] (.7,4) -- (.3,4) node[left]{$d' \le d$};
	\draw[very thick] (.7,.25) -- (.3,.25) node[left]{$\textsf{dim}(\linHull{\mathcal{R}})$};

      \filldraw[color=black!60, fill=black!85,  thick](1,4) circle (0.1);
      \filldraw[color=black!60, fill=black!85, very thick](14,.25) circle (0.1);
    \end{tikzpicture}
  }
  \caption{Tradeoff between number of states and registers: drawing the shape of the set of optimal pairs.}
  \label{fig:tradeoff}
\end{figure}
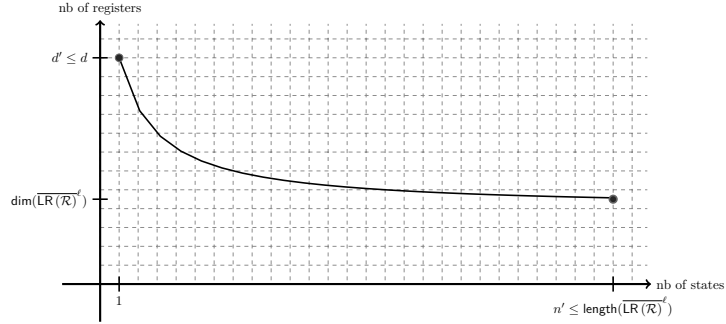

Regarding the proof of the theorem, we know, thanks to
Proposition~\ref{prop:linRepEquivCRA1Stt}, that if one wants to minimize the
number of states of the \CRA, then it is possible to build a \CRA with
a single state, and the minimal number of registers is
the dimension $d'$ of any minimal \WA realizing $f$ (hence $d'\le d$).

On the other side, if one wants to minimize the number of registers, then
Theorem~\ref{thm:mainThm} ensures that the minimal number of
registers is exactly $\textsf{dim}(\linHull{\mathcal{R}})$,
and for this value, the minimal number of states is upper bounded
by $\textsf{length}(\linHull{\mathcal{R}})$.

(This is valid in the affine setting as well thanks to Theorem~\ref{thm:charAff})

\subsection{Details of Remark~\ref{rk:length} } \label{apx:permutMerge}

To illustrate Remark~\ref{rk:length}, we give an example of rational series
which illustrates the expected situation.

\begin{example}
  \label{ex:perm}
  Let $n \in \mathbb{N}$.
  An $n$-dimensional permutation matrix is an $n$ by $n$ matrix that has exactly
  one 1 in each row and each column and zeros everywhere else.
  The set $\permutMatrSet{n}$ of permutation
  matrices together with matrix multiplication form a group of order $n!$
  that is isomorphic to the symmetric group which can be generated using two elements.
  $\linClosure{\permutMatrSet{n}}$ is one-dimensional and has $n!$ irreducible components,
  as it is a union of $n!$ lines.

  It is then possible, using a four-letter alphabet and these permutation matrices,
  to define a $2n$-dimensional \WA with a \LH that is decomposed into the union of
  $n!$ one-dimensional irreducible components and a $(n-1)$-dimensional one,
  such that all the one-dimensional components can be merged together without raising the dimension
  of the invariant.
  We proceed as follows:

  Take two generators of $\permutMatrSet{n}$, \eg let $C_n$ and $T_n$ be the matrices of $\permutMatrSet{n}$
  corresponding to the cycle $(1\, 2\, \dots\, \allowbreak n)$ and the transposition $(1\, 2)$ respectively
  and let $\mathcal{R} = (u,\mu,v)$ be the $2n$-dimensional \WA,
  on $\Sigma = \left\{ a, b, c, d \right\}$, defined by:
  $u = (1\, 2\, \dots\, n\,|\, 0\, \dots\, 0\, 1) \in \vectSet{\mathbb{K}}{2n}$,
  $\mu(a) = \left(
  \begin{array}{c|c}
    C_n & 0   \\
    \hline
    0   & I_n
  \end{array} \right)$, $\mu(b) = \left(
  \begin{array}{c|c}
    T_n & 0   \\
    \hline
    0   & I_n
  \end{array} \right)$, $\mu(c) = \left(
  \begin{array}{c|c}
    0 & 0 \\
    \hline
    0 & M
  \end{array} \right)$ and $\mu(d) = \left(
  \begin{array}{c|c}
    0 & 0  \\
    \hline
    0 & M'
  \end{array} \right)$
  where $M = \left(
  \begin{array}{c|c}
    I_{n-1} & 0 \\
    \hline
    e_1     & 1
  \end{array} \right)$ and $M' = \left(
  \begin{array}{c|c}
    C_{n-1} & 0 \\
    \hline
    0       & 1
  \end{array} \right)$.
  $v \in \matrSet{\mathbb{K}}{n}{1}$ can be arbitrary.

  Note that $\genMono{\mu(a),\mu(b)} = \left\{ \left(
  \begin{array}{c|c}
    P & 0   \\
    \hline
    0 & I_n
  \end{array} \right) \,\middle|\, P \in \permutMatrSet{n} \right\}$ and,
  for all $x \in \vectSet{\mathbb{K}}{n}$
  and $y = (y_1, y_2, \dots, y_{n-1}, 1)  \in \vectSet{\mathbb{K}}{n}$
  \begin{gather*}
  (x \,|\, y)
    \mu(c) = (0 \,|\, (y_1+1), y_2, \dots, y_{n-1}, 1)\\
    (x \,|\, y) \mu(d) = (0 \,|\, y_{n-1}, y_1, \dots, y_{n-2}, 1)
  \end{gather*}

  Then, $\lReachSet{\mathcal{R}} = S_1 \cup S_2$ where
  \begin{gather*}
    S_1 = \left\{ \Bigl( (1\, 2\, \dots\, n) P \,|\, 0\, \dots\, 0\, 1 \Bigr) \,\middle|\,
    P \in \permutMatrSet{n} \right\}\\
    S_2 = \left\{ \Bigl(0\, \dots\, 0 \,|\, l_1\, \dots\, l_{n-1} \, 1 \Bigr)
    \,\middle|\, (l_1\, \dots\, l_{n-1}) \in \mathbb{N}^{n-1} \right\}
  \end{gather*}

  Thus, $\linHull{\mathcal{R}}$ is the union of $\linSpan{S_2}$ and the $n!$ lines,
  going through the origin, directed by the vectors of $S_1$.

  Since $\dim \left( \linSpan{\permutMatrSet{n}} \right) = n-1$
  \ie $\dim \left( \linSpan{S_1} \right) = n-1 = \dim \left( \linSpan{S_2} \right)$,
  we can merge together the one-dimensional irreducible
  components into a single one (which is $\linSpan{S_1}$). We get this way
  a Z-linear invariant
  $I = \linSpan{S_1} \cup \linSpan{S_2}$
  (which is an invariant because for all $a \in \Sigma$ and $i \in \left\{ 1,2 \right\}$,
  there exists $j \in \left\{ 1,2 \right\}$ such that $S_i \mu(a) \subseteq S_j$)
  with the same dimension as $\linHull{\mathcal{R}}$
  but a length of 2, thus reducing the number of states of the corresponding $\CRA$
  to 2 while keeping the same number of registers.
\end{example}

\end{document}